\RequirePackage{fix-cm}

\documentclass[smallcondensed]{svjour3}
\smartqed  
\usepackage[T1]{fontenc}
\usepackage{times,graphicx,amsmath,amssymb,amsfonts,enumerate,cite,float,marvosym,geometry,needspace}
\usepackage{placeins}
\usepackage[colorlinks,bookmarksopen,bookmarksnumbered,citecolor=blue,linkcolor=blue,urlcolor=blue]{hyperref}
\allowdisplaybreaks[2]
\newcommand{\ii}{\mathrm i}
\newcommand{\dd}{\mathrm d}
\newcommand{\ww}{\mathbf w}
\newcommand{\qq}{\mathbf q}
\newcommand{\Id}{\mathrm I}
\newcommand{\Sch}{\mathcal S}
\newcommand{\Cauchy}{\mathcal C}
\newcommand{\diag}{\operatorname{diag}}
\newcommand{\adj}{\operatorname{adj}}

\newcommand{\re}{\operatorname{Re}}
\newtheorem{assumption}{Assumption}

\begin{document}
\begin{sloppypar}

\title{Global well-posedness and long-time asymptotics for the focusing coupled Kundu--Eckhaus equation}

\titlerunning{Global well-posedness and long-time asymptotics for the focusing CKE equation}        

\author{Tianwei Qiu \and Zhen Wang}

\authorrunning{T. Qiu and Z. Wang} 

\institute{T. Qiu \at
              School of Mathematical Sciences, Beihang University, Beijing 100191, China \\
           \and
           Z. Wang \Letter \at
              School of Mathematical Sciences, Beihang University, Beijing 100191, China \\
              \email{wangzmath@163.com}        
}

\date{}

\setlength{\headerboxheight}{0pt}
\renewcommand{\makeheadbox}{}
\maketitle

\begin{abstract}
We prove global well-posedness in the Schwartz space and establish long-time asymptotics for the focusing coupled Kundu--Eckhaus (CKE) equation. For data with no discrete spectrum or real-axis spectral singularities, we analyze the associated matrix Riemann--Hilbert problem by the Deift--Zhou nonlinear steepest descent method. A quantitative local expansion of the full matrix factor near the stationary point identifies the unitary factor in the leading vector term. We combine an estimate for the full matrix coefficient at infinity with exact mass identities. This gives joint error estimates for the gauge-equivalent coupled nonlinear Schr\"odinger solution, the cumulative mass, and the gauge phase. These joint estimates yield an explicit asymptotic formula for CKE radiation with an $O(t^{-1}\log t)$ remainder, uniformly for $x/t$ in compact subsets of $\mathbb R$. We also compute the reflection coefficient and the unitary endpoint factor to visualize the asymptotic solution.

\end{abstract}

\section{Introduction}\label{intro}

Representing a nonlinear evolution equation as the compatibility condition of a Lax pair \cite{Lax1968} forms the foundation of the inverse scattering transform (IST) \cite{GGKM1967,ZS1972}. Within the Riemann--Hilbert (RH) formulation of the IST, the long-time behavior of solutions is governed by an oscillatory matrix factorization problem. The nonlinear steepest descent method, introduced by Deift and Zhou \cite{DeiftZhou}, provides a rigorous framework for extracting sharp long-time asymptotics from such oscillatory RH problems through jump factorizations and contour deformations; see, for instance, its applications to the nonlinear Schr\"odinger (NLS) equation \cite{DeiftZhouNLS}, the Korteweg--de Vries equation \cite{DeiftZhouKdV}, the good Boussinesq equation \cite{DeiftZhouBsq}, and the Fokas--Lenells equation \cite{DeiftZhouFL}.

While the scalar NLS equation models slowly varying wave envelopes in weakly nonlinear dispersive media such as surface water waves and optical fibers \cite{NLS-water1,NLS-optic1,NLS-optic2}, coupled nonlinear Schr\"odinger (CNLS) systems are required to capture the interaction of distinct wave packets or polarization modes \cite{CNLS-optic,CNLS-water}. The canonical integrable model of this type is the Manakov system \cite{Manakov1974}. Writing $\ww=(u,v)$ as a row vector, we consider the Cauchy problem for the focusing CNLS equation
\begin{equation}\label{CNLS}
\begin{aligned}
 &\ii\ww_t+\ww_{xx}+2(\ww\ww^\dagger)\ww=0,\\
 &\ww(x,0)=\ww_0(x)\in\Sch.
\end{aligned}
\end{equation}
Here, $*$ denotes complex conjugation of scalars, and $\dagger$ denotes Hermitian conjugation. Throughout the paper, we write $\Sch=\Sch(\mathbb R;\mathbb C^{1\times2})$ for the Schwartz space of row-vector functions, equipped with its usual Schwartz topology. We use $\Sch(\mathbb R)$ for the scalar Schwartz space. The CNLS equation \eqref{CNLS} admits the following Lax pair \cite{CNLS-DeiftZhou}:
\begin{equation}\label{lax1}
\begin{aligned}
 &\psi_x=(\ii k\sigma+U)\psi,\qquad \psi_t=(2\ii k^2\sigma+\widetilde U)\psi,\\
 &\sigma=\diag(-1,1,1),\qquad U=\ii\begin{pmatrix}0&\ww\\\ww^\dagger&0_{2\times2}\end{pmatrix}, \qquad \widetilde U=2kU+\ii U_x\sigma+\ii U^2\sigma.
\end{aligned}
\end{equation}

Coupled NLS systems admit vector solitons \cite{CNLS-soliton}, breathers \cite{CNLS-breather}, and rogue waves \cite{CNLS-RW1,CNLS-RW2}. In the focusing regime, Ling, Pelinovsky, and Su \cite{CNLS-Stable} established the orbital stability of a family of breather solutions. Geng and Liu \cite{CNLS-DeiftZhou} applied the Deift--Zhou nonlinear steepest descent method to derive long-time asymptotic formulas. Huang, Ling, and Zhang \cite{CNLS-dbar} proved soliton resolution in weighted Sobolev spaces via an improved nonlinear steepest descent method.

Higher-order nonlinear effects in optical pulse propagation motivate integrable extensions of the NLS and Manakov equations. A well-known scalar example is the Kundu--Eckhaus (KE) equation \cite{Kundu,Eckhaus}, which incorporates quintic and derivative nonlinearities while remaining gauge-equivalent to the cubic NLS equation. Its two-component generalization, the coupled Kundu--Eckhaus (CKE) equation, was introduced by Radhakrishnan, Kundu, and Lakshmanan \cite{CKE-ori} to model ultrashort optical pulses in non-Kerr media and also arises in nonlinear electrical transmission networks \cite{CKE-elec}.

We consider the following focusing CKE Cauchy problem:
\begin{equation}\label{CKE}
\begin{aligned}
 \ii q_{j,t}+q_{j,xx}+2S_q q_j+\lambda^2 S_q^2q_j-2\ii\lambda(S_q q_j)_x+2\ii\lambda\Bigl(\sum_{m=1}^2q_m^*q_{m,x}\Bigr)q_j&=0, \quad j=1,2,\\
 S_q=|q_1|^2+|q_2|^2,\qquad \qq(x,0)=\qq_0(x)\in\Sch,&
\end{aligned}
\end{equation}
where $\lambda\in\mathbb R$ is fixed and $\qq=(q_1,q_2)$ is a row vector. Throughout the paper, $\ww_0$ is obtained from $\qq_0$ by the inverse gauge transformation in Proposition~\ref{prop:gauge_equivalence}. The scattering coefficients $a(k)$ and $b(k)$ associated with $\ww_0$ have sizes $2\times2$ and $1\times2$, respectively; their definitions are given in Section~\ref{sec:direct_scattering}. Exact solutions of the focusing CKE equation include solitons \cite{CKE-TianBo,CKE-WDS,CKE-Song}, breathers \cite{CKE-bre}, and rogue waves \cite{CKE-DZ}.

We first establish global well-posedness of the focusing CKE Cauchy problem \eqref{CKE} in the Schwartz space $\Sch$ (Theorem~\ref{thm:global_schwartz}). Combining the gauge equivalence in Proposition~\ref{prop:gauge_equivalence} with the global $H^1$ well-posedness for the one-dimensional cubic CNLS equation \cite[Theorem~1 and Remark~4(i)]{FanelliMontefusco2007}, we obtain a uniform-in-time $H^1$ bound for the CKE solution, propagate higher-order Sobolev regularity and polynomial spatial weights, and prove continuous dependence in the Schwartz topology.

For initial data without discrete spectrum or real-axis spectral singularities, we then establish an explicit long-time asymptotic formula for the focusing CKE equation in Theorem~\ref{thm1} by applying the Deift--Zhou nonlinear steepest descent method \cite{DeiftZhou} to the $3\times3$ RH problem associated with the CNLS scattering data \cite{CNLS-DeiftZhou,CNLS-dbar}. After proving in Lemma~\ref{lem:Schwartz} that the reflection coefficient $\gamma(k)$ belongs to $\Sch$, we derive in Lemma~\ref{lem:endpoint} a quantitative local expansion of the $2\times2$ matrix factor $\delta(k;k_0)$ near the stationary point $k_0=-x/(4t)$ with an $O(|k-k_0|(1+|\log|k-k_0||))$ remainder. This expansion identifies the unitary endpoint factor $E_0(k_0)$ in the leading vector amplitude $\gamma(k_0)E_0(k_0)$, with all bounds holding uniformly for $k_0$ in compact subsets of $\mathbb R$, including points where $\gamma(k_0)=0$.

In Proposition~\ref{eq:main_reduction}, we establish an asymptotic formula for the full coefficient $M_1(x,t)$ of $k^{-1}$ in the RH solution, with an $O_C(t^{-1}\log t)$ remainder. While the $(1,2)$-block of $M_1$ reconstructs the CNLS solution $\ww(x,t)$, its $(1,1)$-entry determines the tail mass. Combining the tail-mass identity with the total-mass trace formula in Proposition~\ref{prop:mass} gives joint estimates for the CNLS solution, the cumulative mass, and the gauge phase. These estimates yield the CKE asymptotic formula with an $O_C(t^{-1}\log t)$ remainder, uniformly for $|x/t|\le C$ with any fixed $C>0$. Finally, in Section~\ref{sec:6}, we compute the reflection coefficient by mapped Chebyshev collocation and evaluate $E_0(k_0)$ by discretizing a regularized singular integral equation \cite{TrogdonOlver2016} to visualize the asymptotic solution.

We write $\mathbb C_\pm=\{k\in\mathbb C:\ \pm\operatorname{Im}k>0\}$.
\begin{assumption}\label{ass:data}
The function $a(k)$ satisfies
\[
 \det a(k)\ne0\quad(k\in\mathbb C_-),\qquad \det a(k)\ne0\quad(k\in\mathbb R).
\]
For $k\in\mathbb R$, $a(k)$ denotes the continuous boundary value of the analytic matrix function from $\mathbb C_-$. It is normalized by $a(k)\to\Id_2$ as $k\to\infty$ in $\overline{\mathbb C}_-$. The two nonvanishing conditions exclude the discrete spectrum and spectral singularities on the real axis, respectively.
\end{assumption}

For $k\in\mathbb R$, define the row-vector reflection coefficient $\gamma(k)=b(k)a(k)^{-1}\in\mathbb C^{1\times2}$, so that $|\gamma(k)|^2=\gamma(k)\gamma(k)^\dagger$. For a matrix $A=(A_{ij})$, we use the Frobenius (Hilbert--Schmidt) norm
\[
 |A|=\sqrt{\sum_{i,j}|A_{ij}|^2}.
\]
For row and column vectors, this coincides with the Euclidean norm. For scalar-, vector-, or matrix-valued functions on $\mathbb R$ or a contour, we use the norms below, with arc-length measure on contours:
\[
 \|f\|_{L^p}=\left(\int |f(\xi)|^p\,|\dd\xi|\right)^{1/p} \quad(1\le p<\infty),\qquad \|f\|_{L^\infty}=\operatorname*{ess\,sup}_{\xi}|f(\xi)|.
\]
For analytic functions on an open domain, $\|f\|_{L^\infty}$
denotes the supremum of $|f|$ over that domain. For a bounded linear operator $\mathcal T$ on $L^2$, we write
\[
 \|\mathcal T\|=\sup_{\|f\|_{L^2}=1}\|\mathcal T f\|_{L^2}.
\]
To state the asymptotic formula, let $\delta(k;k_0)$ denote the normalized $2\times2$ matrix factor defined by \eqref{eq:delta_rh}. Lemma~\ref{lem:delta} shows that $\delta$ exists uniquely and is bounded together with its inverse. Set
\begin{equation}\label{eq:endpoint_data}
\begin{aligned}
 \nu(k_0)&=-\frac{\log(1+|\gamma(k_0)|^2)}{2\pi},\qquad A_0(k_0)=-\frac{\log(1+|\gamma(k_0)|^2)}{2\pi|\gamma(k_0)|^2}\gamma(k_0)^\dagger\gamma(k_0)\quad(\gamma(k_0)\ne0),\\
 E_0(k_0)&=\lim_{k\downarrow k_0}\exp[-\ii A_0(k_0)\log(k-k_0)]\delta(k;k_0).
\end{aligned}
\end{equation}
When $\gamma(k_0)=0$, we define $A_0(k_0)=0$ by continuous extension. We use the principal branch $\log z$ for $z\in\mathbb C\setminus(-\infty,0]$, with $\arg z\in(-\pi,\pi)$; on the positive real axis, this is the real logarithm. The limit in \eqref{eq:endpoint_data} is taken along the real axis with $k>k_0$. Lemma~\ref{lem:endpoint} proves the existence and unitarity of $E_0(k_0)$ and the identity $\det E_0=e^{\chi_0}$, where
\begin{equation}\label{eq:chi0}
 \chi_0(k_0)=\frac1{2\pi\ii}\biggl[\int_{-\infty}^{k_0-1}\frac{\log(1+|\gamma(\xi)|^2)}{\xi-k_0}\,\dd\xi+\int_{k_0-1}^{k_0}\frac1{\xi-k_0}\log\frac{1+|\gamma(\xi)|^2}{1+|\gamma(k_0)|^2}\,\dd\xi\biggr].
\end{equation}
The expression for $\chi_0$ shows that it is purely imaginary. Since $E_0(k_0)$ is unitary, $|\gamma(k_0)E_0(k_0)|=|\gamma(k_0)|$.

\begin{theorem}\label{thm1}
Let $\qq_0\in\Sch$, and let $\qq$ be the global solution of \eqref{CKE} given by Theorem~\ref{thm:global_schwartz}. Suppose that Assumption~\ref{ass:data} holds. For each fixed $C>0$, the following asymptotic formula holds as $t\to+\infty$, uniformly for $|x/t|\le C$:
\begin{equation}\label{eq:CKE_asymptotics}
 \qq(x,t)=\frac{-\ii\Gamma(1+\ii\nu)e^{\pi\nu/2+4\ii k_0^2t+\chi_0(k_0)-\ii\pi/4-\ii\nu\log(8t)}}{2\sqrt{\pi t}}\gamma(k_0)E_0(k_0)e^{\ii\lambda I_0(k_0)}+O_C(t^{-1}\log t),
\end{equation}
where $k_0=-x/(4t)$, $\nu=\nu(k_0)$, $\Gamma$ denotes the gamma function, and
\[
 I_0(k_0)=\frac1\pi\int_{k_0}^{\infty}\log(1+|\gamma(\xi)|^2)\,\dd\xi.
\]
Moreover, the CNLS solution $\ww=(u,v)$ obtained from $\qq$ by the inverse gauge transformation satisfies
\begin{equation}\label{eq:mass_asymptotics}
 \sup_{|x/t|\le C}\left|\int_{-\infty}^{x}(|u|^2+|v|^2)(y,t)\,\dd y-I_0(-x/(4t))\right|=O_C(t^{-1}\log t).
\end{equation}
The constants implicit in $O_C$ may depend on $C$, $\lambda$, and the initial data, but are independent of $(x,t)$ in the stated region.
\end{theorem}

The outline of the paper is as follows. Section~\ref{sec:2} proves the gauge equivalence and global well-posedness in the Schwartz space. Section~\ref{sec:direct_scattering} presents the direct scattering problem and the RH formulation. Section~\ref{sec:4} derives the long-time asymptotic formula for $\ww$ by nonlinear steepest descent. Section~\ref{sec:5} estimates the gauge phase and completes the proof of Theorem~\ref{thm1}. Section~\ref{sec:6} describes the numerical evaluation of the asymptotic solution.

\section{Gauge equivalence and global well-posedness}\label{sec:2}

\subsection{Gauge equivalence}

\begin{proposition}\label{prop:gauge_equivalence}
Fix $T>0$ and consider solutions of \eqref{CNLS} and \eqref{CKE} in $C^1([0,T];\Sch)$. For $\ww=(u,v)$, set $S=\ww\ww^\dagger$. The gauge transformation
\begin{equation}\label{eq:gauge}
 \qq=\ww e^{\ii\Theta}, \qquad \Theta(x,t)=\lambda\int_{-\infty}^x S(y,t)\,\dd y,
\end{equation}
establishes a bijection between solutions of the Cauchy problems \eqref{CNLS} and \eqref{CKE}. Its inverse is given by
\[
 \ww(x,t)=\qq(x,t)\exp\!\left[-\ii\lambda\int_{-\infty}^x|\qq(y,t)|^2\,\dd y\right].
\] 
The gauge transformation and its inverse both preserve $\Sch$ initial data.
\end{proposition}

\begin{proof}
Suppose first that $\qq$ solves \eqref{CKE}, and define $\ww$ by the inverse gauge transformation. Then $\qq\qq^\dagger=\ww\ww^\dagger=S$. Set $B=\ww_x\ww^\dagger-\ww\ww_x^\dagger$. Substituting $\qq=\ww e^{\ii\Theta}$ into \eqref{CKE} and using $\Theta_x=\lambda S$, we obtain
\[
 \ii\ww_t+\ww_{xx}+2S\ww+(\ii\lambda B-\Theta_t)\ww=0.
\]
Since $B^*=-B$ and $\Theta$ is real-valued, the scalar coefficient $\ii\lambda B-\Theta_t$ is real. Multiplying the equation for $\ww$ on the right by $\ww^\dagger$ and its Hermitian conjugate on the left by $\ww$, then subtracting the resulting scalar equations, we obtain $S_t=\ii B_x$. The rapid decay of $\ww$ and $\ww_x$ implies $\lim_{x\to-\infty}B(x,t)=0$. Hence,
\[
 \Theta_t=\lambda\int_{-\infty}^x S_t(y,t)\,\dd y=\ii\lambda B,
\]
so $\ww$ satisfies the differential equation in \eqref{CNLS}. Conversely, if $\ww$ solves \eqref{CNLS}, the local mass conservation law $S_t=\ii B_x$ gives $\Theta_t=\ii\lambda B$. Substitution of $\qq=\ww e^{\ii\Theta}$ then gives the differential equation in \eqref{CKE}. The identity $\qq\qq^\dagger=S$ gives the stated inverse and the corresponding initial data. The phase $\Theta$ is bounded on $\mathbb R\times[0,T]$, and $\partial_x^j\Theta\in C([0,T];\Sch(\mathbb R))$ for every $j\ge1$. All spatial derivatives of $e^{\pm\ii\Theta}$ are therefore bounded on $\mathbb R\times[0,T]$. Since $\Theta_t=\ii\lambda B$ belongs to $C([0,T];\Sch(\mathbb R))$, the Leibniz rule applied to the gauge formulas and their time derivatives shows that both transformations preserve $C^1([0,T];\Sch)$, and hence Schwartz initial data.\qed
\end{proof}

\subsection{Global well-posedness in the Schwartz space} \label{subsec:global_wellposedness}

For row-vector functions, we write $H^m=H^m(\mathbb R;\mathbb C^{1\times2})$ and use the norm
\[
 \|f\|_{H^m}^2=\sum_{r=0}^m\|\partial_x^r f\|_{L^2}^2, \qquad m=0,1,2,\ldots.
\]

\begin{theorem}\label{thm:global_schwartz}
Let $\qq_0\in\Sch$. Then the Cauchy problem \eqref{CKE} admits a unique global solution
\[
 \qq\in C^1([0,\infty);\Sch).
\]
The solution satisfies the uniform bound
\[
 \sup_{t\ge0}\|\qq(t)\|_{H^1}<\infty.
\]
Moreover, for every $T>0$, the data-to-solution map
\[
 \begin{aligned}
 \Sch&\longrightarrow C^1([0,T];\Sch),\\
 \qq_0&\longmapsto\qq|_{[0,T]}
 \end{aligned}
\]
is continuous with respect to the Schwartz topology on the initial data and the induced $C^1$ topology on the solutions.
\end{theorem}

\begin{proof}
For $f\in\Sch$, set
\[
 (\mathcal G_\lambda f)(x)=f(x)\exp\!\left[\ii\lambda\int_{-\infty}^x|f(y)|^2\,\dd y\right].
\]
By Proposition~\ref{prop:gauge_equivalence}, $\mathcal G_\lambda$ and $\mathcal G_{-\lambda}$ are inverse maps on $\Sch$. Let $\ww_0=\mathcal G_{-\lambda}\qq_0$. With $n=1$, $p=1$, and $\beta=1$ in \cite{FanelliMontefusco2007}, the rescaling $(\phi,\psi)=\sqrt2(u,v)$ gives \eqref{CNLS} without changing the space or time variable. Theorem~1 and Remark~4(i) of that work therefore give a unique global solution $\ww\in C([0,\infty);H^1)$. Its mass and energy
are conserved \cite[Section~2]{FanelliMontefusco2007}:
\[
 \begin{aligned}
 &\|\ww(t)\|_{L^2}^2=\int_{\mathbb R}|\ww_0|^2\,\dd x=:\mathcal M_1,\\
 &\int_{\mathbb R}(|\ww_x(t)|^2-|\ww(t)|^4)\,\dd x=\int_{\mathbb R}(|\ww_{0,x}|^2-|\ww_0|^4)\,\dd x=:\mathcal M_2.
 \end{aligned}
\]
The one-dimensional Gagliardo--Nirenberg inequalities give
\[
 \|f\|_{L^4}^4\le C_4\|f\|_{L^2}^3\|f_x\|_{L^2},\qquad\|f\|_{L^6}^3\le C_6\|f\|_{L^2}^2\|f_x\|_{L^2}.
\]
For row-vector functions, these follow by applying the scalar inequalities to $|f|$ and using $|\partial_x|f||\le|f_x|$ almost everywhere. Writing $X(t)=\|\ww_x(t)\|_{L^2}$, we obtain
\[
 \mathcal M_2\ge X(t)^2-C_4\mathcal M_1^{3/2}X(t)\ge\tfrac12X(t)^2-\tfrac12C_4^2\mathcal M_1^3.
\]
Thus,
\[
 X(t)\le K:=\bigl(2|\mathcal M_2|+C_4^2\mathcal M_1^3\bigr)^{1/2},\qquad\sup_{t\ge0}\|\ww(t)\|_{H^1}\le\mathcal M_1^{1/2}+K.
\]

Fix $T>0$ and write $\mathcal N(f)=(ff^\dagger)f$.
For every integer $m\ge1$, the Sobolev product estimate gives
\[
 \|\mathcal N(f)\|_{H^m}\le C_m\|f\|_{L^\infty}^2\|f\|_{H^m};
\]
see \cite[Section~3.3]{Tao2006}. A contraction argument in $C([0,\tau];H^m)$, for sufficiently small $\tau>0$, gives a local solution, which agrees with $\ww$ by uniqueness in $H^1$. The Duhamel formula is
\[
 \ww(t)=e^{\ii t\partial_x^2}\ww_0+2\ii\int_0^t e^{\ii(t-s)\partial_x^2}\mathcal N(\ww(s))\,\dd s.
\]
Here, $h(t)=e^{\ii t\partial_x^2}f$ solves $\ii h_t+h_{xx}=0$ with $h(0)=f$. The operator acts componentwise as the Fourier multiplier $e^{-\ii t\xi^2}$ and is unitary on every $H^m$. On the interval of existence, the product estimate, the uniform $H^1$ bound, and Gronwall's inequality give
\[
 \|\ww(t)\|_{H^m}\le\|\ww_0\|_{H^m}\exp\!\left[C_m(\mathcal M_1^{1/2}+K)^2t\right].
\]
The local $H^m$ solution therefore extends to $[0,T]$, with the corresponding bound uniform on this interval.

We next propagate the spatial weights. For nonnegative integers $j,m$, set
\[
 Y_{j,m}(t)=\sum_{r=0}^m\|x^j\partial_x^r\ww(t)\|_{L^2}.
\]
We prove by induction on $j$, simultaneously for all $m$, that $x^j\partial_x^r\ww\in C([0,T];L^2)$ for $0\le r\le m$. The case $j=0$ follows from the Sobolev estimates. For $j\ge1$, the commutator identity
\[
 [\partial_x^2,x^j]=2jx^{j-1}\partial_x+j(j-1)x^{j-2}
\]
gives
\[
 (\ii\partial_t+\partial_x^2)(x^j\partial_x^r\ww)=-2x^j\partial_x^r\mathcal N(\ww)+2jx^{j-1}\partial_x^{r+1}\ww+j(j-1)x^{j-2}\partial_x^r\ww.
\]
By Leibniz's rule, placing the weight on one factor and estimating the other two in $L^\infty$, we have
\[
 \sum_{r=0}^m\|x^j\partial_x^r\mathcal N(\ww)\|_{L^2}\le C_m\|\ww\|_{W^{m,\infty}}^2Y_{j,m}.
\]
The Sobolev bounds imply $B_{m,T}:=\sup_{0\le t\le T}\|\ww(t)\|_{W^{m,\infty}}<\infty$. Applying Duhamel's formula to the displayed equation for $x^j\partial_x^r\ww$ and using the $L^2$ unitarity of $e^{\ii t\partial_x^2}$, we obtain
\[
 Y_{j,m}(t)\le Y_{j,m}(0)+C_{j,m}B_{m,T}^2\int_0^tY_{j,m}(s)\,\dd s+C_{j,m}\int_0^t\bigl(Y_{j-1,m+1}(s)+Y_{j-2,m}(s)\bigr)\,\dd s,\qquad 0\le t\le T,
\]
where the term with weight $j-2$ is omitted when $j=1$.

To justify this estimate, replace $x^j$ by $x^j\chi(x/R)$, where $R\ge1$, $\chi\in C_c^\infty(\mathbb R)$, $0\le\chi\le1$, $\chi=1$ on $[-1,1]$, and $\chi=0$ outside $[-2,2]$. The first derivative of this weight is bounded by $C_j|x|^{j-1}$; for $j\ge2$, its second derivative is bounded by $C_j|x|^{j-2}$. For $j=1$, the second derivative is $O(R^{-1})$ and is supported where $R\le|x|\le2R$. Thus, the commutator terms are controlled by the lower-weight induction hypothesis, with an additional term $CR^{-1}\|\ww\|_{H^m}$ when $j=1$. The nonlinear estimate holds with the same cutoff on the weighted factor. Gronwall's inequality gives bounds independent of $R$, and Fatou's lemma gives the corresponding bounds for $Y_{j,m}$. The truncated weighted derivatives converge in $L^2$ at each fixed time. All terms on the right-hand side of their equations converge in $L^1([0,T];L^2)$ by dominated convergence; the additional term for $j=1$ tends to zero. Passing to the limit in Duhamel's formula gives an $L^2$-continuous representative of each $x^j\partial_x^r\ww$, completing the induction.

The weighted $L^2$ seminorms $\|x^j\partial_x^r f\|_{L^2}$ generate the Schwartz topology. Sobolev embedding applied to $x^j\partial_x^r f$ bounds the weighted supremum seminorms by weighted $L^2$ seminorms. Conversely, each weighted $L^2$ seminorm is bounded by a finite sum of weighted supremum seminorms. Hence, $\ww\in C([0,T];\Sch)$. Since $f\mapsto\ii f_{xx}+2\ii\mathcal N(f)$ is continuous from $\Sch$ to itself, the integral form of \eqref{CNLS} gives $\ww\in C^1([0,T];\Sch)$.

Set $\qq=\mathcal G_\lambda\ww$. Proposition~\ref{prop:gauge_equivalence} gives a solution of \eqref{CKE} in $C^1([0,\infty);\Sch)$ with initial value $\qq_0$. The identities $|\qq|=|\ww|$ and $\qq_x-\ii\lambda|\qq|^2\qq=e^{\ii\Theta}\ww_x$ imply
\[
 \int_{\mathbb R}|\qq|^2\,\dd x=\mathcal M_1,\qquad \int_{\mathbb R}\left(\bigl|\qq_x-\ii\lambda|\qq|^2\qq\bigr|^2-|\qq|^4\right)\,\dd x=\mathcal M_2.
\]
The $L^6$ estimate above gives
\[
 \|\qq_x(t)\|_{L^2}\le\|\ww_x(t)\|_{L^2}+|\lambda|\|\ww(t)\|_{L^6}^3\le(1+C_6|\lambda|\mathcal M_1)K.
\]
Together with conservation of mass, this yields
\[
 \sup_{t\ge0}\|\qq(t)\|_{H^1}\le\mathcal M_1^{1/2}+(1+C_6|\lambda|\mathcal M_1)K<\infty.
\]
Any other CKE solution in $C^1([0,\infty);\Sch)$ with initial value $\qq_0$ is mapped by $\mathcal G_{-\lambda}$ to a CNLS solution with initial value $\ww_0$. Uniqueness in $C([0,T];H^1)$ for every $T>0$ therefore proves uniqueness for CKE.

For continuous dependence, let $\ww_{0,n}\to\ww_0$ in $\Sch$, and let $\ww_n$ and $\ww$ be the corresponding CNLS solutions. Their masses, energies, and initial $H^m$ norms are uniformly bounded, so the preceding estimates give uniform $H^m$ bounds on $[0,T]$. Set $\boldsymbol z_n=\ww_n-\ww$. The identity
\[
 \mathcal N(\ww_n)-\mathcal N(\ww)=\boldsymbol z_n\ww_n^\dagger\ww_n+\ww\boldsymbol z_n^\dagger\ww_n+\ww\ww^\dagger\boldsymbol z_n
\]
allows the weight in each differentiated product to be placed on the factor containing $\boldsymbol z_n$. Consequently,
\[
 \sum_{r=0}^m\|x^j\partial_x^r(\mathcal N(\ww_n)-\mathcal N(\ww))\|_{L^2}\le C_{m,T}\sum_{r=0}^m\|x^j\partial_x^r\boldsymbol z_n\|_{L^2},
\]
with $C_{m,T}$ independent of $n$. Writing $Z_{j,m}^{(n)}(t)=\sum_{r=0}^m\|x^j\partial_x^r\boldsymbol z_n(t)\|_{L^2}$, the same weighted estimate gives
\[
 Z_{j,m}^{(n)}(t)\le Z_{j,m}^{(n)}(0)+C_{j,m,T}\int_0^t\bigl(Z_{j,m}^{(n)}(s)+Z_{j-1,m+1}^{(n)}(s)+Z_{j-2,m}^{(n)}(s)\bigr)\,\dd s,
\]
with terms of negative weight omitted. We first apply Gronwall's inequality with $j=0$. Induction on $j$, carried out simultaneously for all derivative orders $m$, then gives $\sup_{0\le t\le T}Z_{j,m}^{(n)}(t)\to0$ for every $j,m$. Thus, $\ww_n\to\ww$ in $C([0,T];\Sch)$. The equation $\ww_t=\ii\ww_{xx}+2\ii\mathcal N(\ww)$ gives convergence of the time derivatives in the same space.

Finally, for $f,g\in\Sch$, let $\Theta_f(x)=\lambda\int_{-\infty}^x|f(y)|^2\,\dd y$. Then
\[
 \begin{aligned}
 &\|\Theta_f-\Theta_g\|_{L^\infty} \le|\lambda|(\|f\|_{L^2}+\|g\|_{L^2})\|f-g\|_{L^2},\\
 &\partial_x^r\Theta_f=\lambda\partial_x^{r-1}|f|^2,\qquad r\ge1.
 \end{aligned}
\]
These estimates and Leibniz's rule give continuity of $\mathcal G_{\pm\lambda}$ in the Schwartz topology, with estimates uniform on bounded sets. If $\qq_{0,n}\to\qq_0$ in $\Sch$, their inverse gauge transforms therefore converge in $\Sch$. The continuous dependence for CNLS and the gauge estimates yield $\qq_n\to\qq$ in $C([0,T];\Sch)$. Solving \eqref{CKE} for $\qq_t$ expresses the time derivative as a continuous map of $\qq$ from $\Sch$ to itself. Hence, $\qq_n\to\qq$ in $C^1([0,T];\Sch)$, as required. \qed
\end{proof}

\begin{remark}\label{rem:cke_wellposedness}
The derivative relation $\qq_x=e^{\ii\Theta}(\ww_x+\ii\lambda|\ww|^2\ww)$ shows that passing from the CNLS equation \eqref{CNLS} to the CKE equation \eqref{CKE} at the $H^1$ level requires controlling the cubic term $|\ww|^2\ww$ in $L^2$. The conserved CNLS mass $\mathcal M_1$ and energy $\mathcal M_2$, combined with the one-dimensional Gagliardo--Nirenberg inequality, provide this bound uniformly in time. The weighted estimates for solutions and their differences, together with continuity of the gauge maps $\mathcal G_{\pm\lambda}$, give persistence of Schwartz regularity and continuous dependence in $C^1([0,T];\Sch)$. Moreover, because the one-dimensional cubic CNLS equation is $L^2$-subcritical, Theorem~\ref{thm:global_schwartz} holds for arbitrary initial data $\qq_0\in\Sch$ without any smallness condition.
\end{remark}

\section{Direct scattering and the RH problem}\label{sec:direct_scattering}

\subsection{Jost solutions and scattering data}

For the Lax pair \eqref{lax1}, set $\Phi(x,t;k)=kx+2k^2t$ and $\mu=\psi e^{-\ii\Phi\sigma}$. The spatial equation reads
\[
 \mu_x=\ii k[\sigma,\mu]+U\mu,
\]
and, for real $k$, the Jost solutions satisfy the Volterra equations
\begin{equation}\label{eq:Volterra}
 \mu_\pm(x,t;k)=\Id_3+\int_{\pm\infty}^{x}e^{\ii k(x-y)\widehat\sigma}U(y,t)\mu_\pm(y,t;k)\,\dd y,\quad e^{z\widehat\sigma}X=e^{z\sigma}Xe^{-z\sigma}.
\end{equation}

The analyticity and normalization properties of the Jost solutions stated below follow from \cite[Section~2]{CNLS-DeiftZhou}.

\begin{proposition}\label{prop:Jost}
For fixed $x,t$, write $\mu_\pm=(\mu_{\pm,L},\mu_{\pm,R})$, where $\mu_{\pm,L}$ is the first column and $\mu_{\pm,R}$ consists of the last two columns. The matrices $(\mu_{-,L},\mu_{+,R})$ and $(\mu_{+,L},\mu_{-,R})$ are analytic for $k\in\mathbb C_+$ and $k\in\mathbb C_-$, respectively, and extend continuously to the real axis. Moreover, $(\mu_{-,L},\mu_{+,R})=\Id_3+O(k^{-1})$ as $k\to\infty$ in $\overline{\mathbb C}_+$, and $(\mu_{+,L},\mu_{-,R})=\Id_3+O(k^{-1})$ as $k\to\infty$ in $\overline{\mathbb C}_-$. On the real axis, $\det\mu_\pm=1$ and $\mu_\pm(k)^\dagger=\mu_\pm(k)^{-1}$.
\end{proposition}

For the Jost solutions of Proposition~\ref{prop:Jost}, we define the scattering matrix by
\[
 \mu_-=\mu_+ e^{\ii\Phi\widehat\sigma}s(k),\qquad \det s=1.
\]
The scattering matrix $s(k)$ is independent of $x$ and $t$. Its $1+2$ block decomposition is
\begin{equation}\label{eq:scattering_blocks}
 s(k)=\begin{pmatrix}
 (\det a(k))^*&b(k)\\
 -\adj(a(k)^\dagger)b(k)^\dagger&a(k)
 \end{pmatrix},\qquad k\in\mathbb R.
\end{equation}
Here, $\adj(A)$ denotes the adjugate of $A$, so that $\adj(A)A=(\det A)\Id$. We also recall the following properties from \cite[Section~2]{CNLS-DeiftZhou}.

\begin{proposition}\label{prop:symmetry}
The scattering matrix $s(k)$ is unitary for $k\in\mathbb R$. The representation \eqref{eq:scattering_blocks} follows from $s^{-1}=s^\dagger$ and $\det s=1$. The function $a(k)$ is analytic in $\mathbb C_-$, and the function $(\det a(k^*))^*$ is analytic in $\mathbb C_+$. The coefficient $b(k)$ and the scattering matrix $s(k)$ are defined for $k\in\mathbb R$. Under Assumption~\ref{ass:data}, $a(k)^{-1}$ is uniformly bounded on $\mathbb R$.
\end{proposition}

For the Jost solutions, the scattering matrix, and the RH matrices, the subscripts $11$, $12$, $21$, and $22$ refer to the $1\times1$, $1\times2$, $2\times1$, and $2\times2$ blocks in the $1+2$ decomposition, respectively. At $t=0$, with $\mu_-$ evaluated at $(x,0;k)$, the scattering coefficients satisfy
\begin{equation}\label{abk}
 a(k)=\Id_2+\ii\int_{\mathbb R}\ww_0^\dagger\mu_{-,12}\,\dd x, \qquad b(k)=\ii\int_{\mathbb R}e^{2\ii kx}\ww_0\mu_{-,22}\,\dd x.
\end{equation}
Finite-order regularity of the direct scattering transform in weighted Sobolev spaces is proved in \cite[Proposition~3.1 and Remark~3.2]{Liu-direct}. For rational approximation, we need the following Schwartz regularity of the reflection coefficient.

\begin{lemma}\label{lem:Schwartz}
Let $\ww_0\in\Sch$, and suppose that $a(k)$ satisfies Assumption~\ref{ass:data}. Then $\gamma\in\Sch$.
\end{lemma}

\begin{proof}
Iterating the Volterra equation for $\mu_-$ in \eqref{eq:Volterra} and substituting the resulting Neumann series into \eqref{abk} gives a series for $b$. Only terms containing an odd number $n$ of factors of $U$ contribute; their phase is $e^{2\ii k\alpha_n}$, where $\alpha_n=\sum_{j=1}^{n}(-1)^{j+1}x_j$ on the ordered domain $x_1<\cdots<x_n$. Differentiating $m$ times with respect to $k$ introduces a factor proportional to $\alpha_n^m$. For odd $n$, the vector field $\sum_{j=1}^n\partial_{x_j}$ is tangent to the boundaries $x_j=x_{j+1}$ of the ordered integration domain and satisfies $(\sum_{j=1}^n\partial_{x_j})\alpha_n=1$. Integrating by parts $N$ times gives decay of order $|k|^{-N}$ for $|k|\ge1$, with no contributions from the interior boundaries; Schwartz decay eliminates the boundary terms at infinity. The function $(1+|x|)^m\sum_{j=0}^N|\partial_x^jU(x,0)|$ provides a common integrable bound for the factors in the differentiated integrands. After replacing each factor by this bound, the product is symmetric in $x_1,\ldots,x_n$, and integration over the ordered domain contributes the factor $1/n!$. Thus,
\[
 \sup_{k\in\mathbb R}|k^N\partial_k^m b(k)|\le\sum_{\substack{n\ge1\\n\ \text{odd}}}\frac{n^{m+N}C_{m,N}^{n}}{n!}<\infty,\qquad m,N=0,1,2,\ldots.
\]
Here, $C_{m,N}$ depends on finitely many Schwartz seminorms. For $|k|\le1$, the differentiated series for $b$ is bounded without integration by parts. Applying the same absolute majorant to the even terms, also without integration by parts, gives $\sup_{k\in\mathbb R}|\partial_k^m a(k)|<\infty$ for every nonnegative integer $m$. Assumption~\ref{ass:data} and the limit $a(k)\to\Id_2$ imply that $a^{-1}$ is bounded on $\mathbb R$. Differentiating the identity $aa^{-1}=\Id_2$ yields bounds for all derivatives of $a^{-1}$. The product rule applied to $\gamma=ba^{-1}$ proves the claim.\qed
\end{proof}

\subsection{The original RH problem}

Using Proposition~\ref{prop:symmetry}, we define $M$ from the Jost solutions by
\begin{equation}\label{eq:Jost_RH}
M(k)=
\begin{cases}
\bigl(\mu_{-,L}[(\det a(k^*))^*]^{-1},\ \mu_{+,R}\bigr),
& k\in\mathbb C_+,\\
\bigl(\mu_{+,L},\ \mu_{-,R}a^{-1}\bigr),
& k\in\mathbb C_-.
\end{cases}
\end{equation}
We denote the boundary values from $\mathbb C_\pm$ by $M_\pm$. The last two columns of the scattering relation give
\[
 \mu_{-,R}a^{-1}=\mu_{+,R}+e^{-2\ii\Phi}\mu_{+,L}\gamma.
\]
The first column of the scattering relation, together with the cofactor identity in \eqref{eq:scattering_blocks}, yields
\[
 M_{+,L}=(1+|\gamma|^2)M_{-,L}-e^{2\ii\Phi}M_{-,R}\gamma^\dagger.
\]
Thus, $M$ satisfies
\begin{equation}\label{eq:original_RH}
 M_+=M_-J,\qquad J=\begin{pmatrix}1+\gamma\gamma^\dagger&-e^{-2\ii\Phi}\gamma\\
                 -e^{2\ii\Phi}\gamma^\dagger&\Id_2\end{pmatrix}, \qquad M(\infty)=\Id_3.
\end{equation}

The solvability and reconstruction results are given in \cite[Theorem~2]{CNLS-DeiftZhou}. We also record the coefficient identities from the spatial Lax equation in the row-vector convention used here.

\begin{proposition}\label{recon1}
Under Assumption~\ref{ass:data}, the RH problem \eqref{eq:original_RH} has a unique solution $M(x,t;k)$ analytic in $\mathbb C\setminus\mathbb R$. It is normalized by $M(\infty)=\Id_3$, has no poles, and extends continuously to $\mathbb R$ from each half-plane. Its boundary values are bounded and satisfy $M_\pm-\Id_3\in L^2(\mathbb R)$. It satisfies $\det M(k)=1$ and $M(k)^{-1}=M(k^*)^\dagger$. If $M=\Id_3+M_1/k+O(k^{-2})$, then
\begin{equation}\label{eq:reconstruction}
 \ww=2(M_1)_{12},\qquad (M_1)_{21}=-\tfrac12\ww^\dagger, \qquad \partial_x(M_1)_{11}=-\tfrac\ii2 S.
\end{equation}
The function $\ww$ reconstructed by \eqref{eq:reconstruction} solves the Cauchy problem \eqref{CNLS}.
\end{proposition}

By uniqueness for the CNLS Cauchy problem, the reconstructed solution coincides with the CNLS solution used in the proof of Theorem~\ref{thm:global_schwartz}.

\section{The nonlinear steepest descent method}\label{sec:4}

In this section, we analyze the RH problem \eqref{eq:original_RH} by the nonlinear steepest descent method. We first factorize the jump matrix and deform the contour through the stationary point $k_0=-x/(4t)$. The local expansion of $\delta(k;k_0)$ determines the jumps of the parabolic cylinder model on the resulting cross. Comparing the deformed RH problem with this model yields an asymptotic formula for the full coefficient $M_1(x,t)$ of $k^{-1}$ with an $O_C(t^{-1}\log t)$ remainder, which is used in Section~\ref{sec:5} to recover both the CNLS solution and the gauge phase.

Throughout this section, $t>0$ and $\theta(k)=\Phi(k)/t=2(k-k_0)^2-2k_0^2$, where $k_0=-x/(4t)$ is the stationary point. We write $\nu=\nu(k_0)$ and $\chi_0=\chi_0(k_0)$. All estimates are uniform for $k_0$ in any fixed compact interval. For an oriented contour, the $+$-side (resp.~$-$-side) is the region to the left (resp.~right) of the direction of traversal.

\subsection{Triangular factorizations of the jump matrix}

With $\mathsf U(f)=\begin{pmatrix}1&f\\0&\Id_2\end{pmatrix}$ and $\mathsf L(h)=\begin{pmatrix}1&0\\h&\Id_2\end{pmatrix}$, the real-axis jump admits the factorizations
\begin{equation}\label{eq:factorizations}
\begin{aligned}
 J&=\mathsf U(-e^{-2\ii t\theta}\gamma) \mathsf L(-e^{2\ii t\theta}\gamma^\dagger),\\
 J&=\mathsf L\left(-\frac{e^{2\ii t\theta}\gamma^\dagger}{1+|\gamma|^2}\right)\diag\bigl(1+|\gamma|^2,(\Id_2+\gamma^\dagger\gamma)^{-1}\bigr)\mathsf U\left(-\frac{e^{-2\ii t\theta}\gamma}{1+|\gamma|^2}\right).
\end{aligned}
\end{equation}
The lower-right block in the second factorization follows from $(\Id_2+\gamma^\dagger\gamma)^{-1}=\Id_2-\gamma^\dagger\gamma/(1+|\gamma|^2)$.

To remove the block-diagonal factor on $(-\infty,k_0)$ in the second factorization of \eqref{eq:factorizations}, we introduce the $2\times2$ matrix factor $\delta(k)=\delta(k;k_0)$ through the normalized RH problem with the jump matrix acting by left multiplication:
\begin{equation}\label{eq:delta_rh}
 \delta_+(k)=(\Id_2+\gamma(k)^\dagger\gamma(k))\delta_-(k)\quad(k<k_0),\qquad \delta(\infty)=\Id_2.
\end{equation}
We seek a solution analytic in $\mathbb C\setminus(-\infty,k_0]$, with $\delta_\pm-\Id_2\in L^2(\mathbb R)$, such that both $\delta$ and $\delta^{-1}$ are bounded on this domain, including near the endpoint $k_0$.

The solvability and symmetry of \eqref{eq:delta_rh} are established in \cite[Section~3.1]{CNLS-DeiftZhou}. The boundary identities and inverse bound below follow from the jump relation and symmetry.

\begin{lemma}\label{lem:delta}
The RH problem \eqref{eq:delta_rh} has a unique solution, and
\begin{equation}\label{eq:delta_sym_bounds}
 \delta(k)^{-1}=\delta(k^*)^\dagger,\qquad \|\delta^{\pm1}\|_{L^\infty(\mathbb C\setminus(-\infty,k_0])}+\|d^{\pm1}\|_{L^\infty(\mathbb C\setminus(-\infty,k_0])}\le K,
\end{equation}
where
\begin{equation}\label{eq:det_delta}
 d(k)=\det\delta(k)=\exp\left[\frac1{2\pi\ii}\int_{-\infty}^{k_0}\frac{\log(1+|\gamma(\xi)|^2)}{\xi-k}\,\dd\xi\right].
\end{equation}
The same bounds hold for both boundary values away from the endpoint. On the cut,
\[
 \delta_+\delta_+^\dagger=\Id_2+\gamma^\dagger\gamma,\qquad \delta_-\delta_-^\dagger=(\Id_2+\gamma^\dagger\gamma)^{-1};
\]
for $k>k_0$, $\delta(k)$ is unitary.
\end{lemma}

By the maximum modulus principle, the boundary identities ensure that the $L^\infty$ bounds in \eqref{eq:delta_sym_bounds} depend only on $\|\gamma\|_{L^\infty}$, and are therefore independent of $k_0\in\mathbb R$. The boundedness of $\delta_-$ and the Schwartz decay of $\gamma$ imply $\xi^j\gamma^\dagger\gamma\delta_-\mathbf1_{(-\infty,k_0)}\in L^1$ for $j=0,1$. The additive Cauchy representation therefore gives
\begin{equation}\label{eq:delta_expansion}
 \delta(k)=\Id_2+\frac{\delta^{[1]}}k+O(k^{-2}),\qquad \delta^{[1]}=-\frac1{2\pi\ii}\int_{-\infty}^{k_0}\gamma^\dagger\gamma\delta_-\,\dd\xi.
\end{equation}
The expansion holds in closed angular sectors separated from the real axis. The remainder is controlled by the first absolute moment of the jump density.

With the principal branch specified after \eqref{eq:endpoint_data}, the scalar factor satisfies
\begin{equation}\label{eq:d_local}
 d(k)=(k-k_0)^{\ii\nu}e^{\chi_0(k_0)}[1+O(|k-k_0|(1+|\log|k-k_0||))]
\end{equation}
as $k\to k_0$ along the rays $\arg(k-k_0)=\pm\pi/4,\pm3\pi/4$ or along the real axis with $k>k_0$. Equation \eqref{eq:chi0} and this expansion are obtained by subtracting the endpoint value of the logarithm on $(k_0-1,k_0)$ in \eqref{eq:det_delta}.

Define
\begin{equation}\label{eq:Delta_rho}
 \Delta=\diag(d,\delta^{-1}),\qquad \rho(k)=\begin{cases}-\gamma(k)/(1+|\gamma(k)|^2),&k<k_0,\\
          \gamma(k),&k>k_0,
          \end{cases}
 \qquad M^\Delta=M\Delta^{-1}.
\end{equation}
By \eqref{eq:det_delta} and  \eqref{eq:delta_expansion}, in closed angular sectors separated from the real axis,
\[
 \Delta(k)=\Id_3+\frac{\Delta_1}{k}+O(k^{-2}),\qquad \Delta_1=\diag(d^{[1]},-\delta^{[1]}),\qquad d(k)=1+\frac{d^{[1]}}{k}+O(k^{-2}).
\]
Following the transformation \eqref{eq:Delta_rho}, we reverse the orientation of the semi-infinite interval $(k_0,\infty)$ as depicted in Fig.~\ref{fig1:contour}(b). This orientation reversal interchanges the $\pm$ boundary values, so that the jump matrix is replaced by its inverse. The resulting jump matrix is
\begin{equation}\label{eq:M_Delta_rh}
 J^\Delta=\mathsf L\left(e^{2\ii t\theta}\delta_-^{-1}\rho^\dagger/d_-\right)\mathsf U\left(e^{-2\ii t\theta}d_+\rho\delta_+\right).
\end{equation}
On $(-\infty,k_0)$, the jump relation $\delta_+=(\Id_2+\gamma^\dagger\gamma)\delta_-$ and the identity $d_+=(1+|\gamma|^2)d_-$ give
\[
\begin{aligned}
 &\Delta_-^{-1}\Delta_+=\diag\bigl(1+|\gamma|^2,(\Id_2+\gamma^\dagger\gamma)^{-1}\bigr),\\
 &\Delta_-\diag\bigl(1+|\gamma|^2,(\Id_2+\gamma^\dagger\gamma)^{-1}\bigr)\Delta_+^{-1}=\Id_3.
\end{aligned}
\]
Consequently, the second factorization in \eqref{eq:factorizations} yields
\[
 \Delta_-J\Delta_+^{-1}=\mathsf L\left(-\frac{e^{2\ii t\theta}\delta_-^{-1}\gamma^\dagger}{(1+|\gamma|^2)d_-}\right)\mathsf U\left(-\frac{e^{-2\ii t\theta}d_+\gamma\delta_+}{1+|\gamma|^2}\right).
\]
On the reversed right half-line, $\Delta$ has no jump, and
\[
 \Delta J^{-1}\Delta^{-1}=\mathsf L(e^{2\ii t\theta}\delta^{-1}\gamma^\dagger/d)\mathsf U(e^{-2\ii t\theta}d\gamma\delta).
\]
The two expressions give \eqref{eq:M_Delta_rh} with the respective definitions of $\rho$.

\begin{figure}[t]
 \centering\includegraphics[width=.95\linewidth]{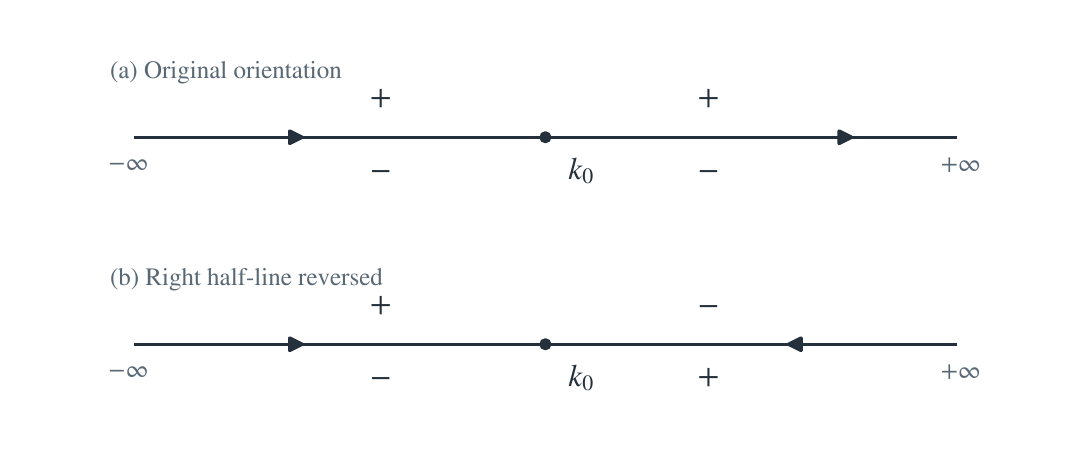}
 \caption{Orientation of the real contour: (a) the original left-to-right orientation; (b) the orientation after reversing $(k_0,\infty)$. In (b), both half-lines are oriented toward $k_0$. The $+$-side lies to the left of each oriented segment.}\label{fig1:contour}
\end{figure}

\subsection{Decomposition of \texorpdfstring{$\rho(k)$}{rho(k)}}

Let $L$ consist of the rays from $k_0$ at angles $3\pi/4$ and $-\pi/4$, and let $L^*$ consist of their conjugates. All four rays are oriented toward $k_0$. By Lemma~\ref{lem:Schwartz}, the restriction of $\rho$ to each half-line is smooth up to $k_0$ and admits the Taylor expansion used below.

\begin{lemma}\label{rational-de}
For each integer $l\ge1$, there exists a decomposition $\rho=R+h_1+h_2$ on each half-line, where $R$ is rational and has no poles in the closed sector between that half-line and $L$, and $h_2$ is analytic in the corresponding open sector. The functions $h_1$ and $h_2$ may depend on $t$ and satisfy
\[
 |e^{-2\ii t\theta}h_1|\le K_l t^{-l}(1+|k-k_0|^2)^{-1}\quad(k\in\mathbb R),\qquad|e^{-2\ii t\theta}h_2|\le K_l t^{-l}(1+|k-k_0|^2)^{-1}\quad(k\in L).
\]
\end{lemma}

\begin{proof}
We first treat the right half-line, where $\tau=k-k_0\ge0$. Choose an integer $p\ge2l+2$ and set $m=3p+1$. Define $R(k_0+\tau)=(\tau-\ii)^{-(m+5)}P_m(\tau)$, where $P_m(\tau)$ is the Taylor polynomial of degree $m$ of the function $(\tau-\ii)^{m+5}\rho(k_0+\tau)$ at $\tau=0$. Then $h=\rho-R=\tau^{m+1}G(\tau)/(\tau-\ii)^{m+5}$ near zero, where $G$ is smooth with bounded derivatives. Define $\theta_0=-2k_0^2$ and
\[
 f(\theta)=\begin{cases}
 (\tau-\ii)^{p+2}\tau^{-p}h(k_0+\tau),&\theta\ge\theta_0,\quad
                                   \theta=\theta_0+2\tau^2,\\
 0,&\theta<\theta_0.
 \end{cases}
\]
Since $\partial_\theta=(4\tau)^{-1}\partial_\tau$ and $f=\tau^{2p+2}G(\tau)/(\tau-\ii)^{2p+4}$ near zero,
\[
 |\partial_\theta^jf|\le K \tau^{2p+2-2j}, \qquad 0<\tau\le1,\quad 0\le j\le p+1.
\]
For $\tau\ge1$, $h^{(j)}=O(\tau^{-5-j})$, so $\partial_\theta^jf=O(\tau^{-3-2j})$. These bounds are square-integrable with respect to $\dd\theta=4\tau\,\dd \tau$. Since $\partial_\theta^j f(\theta_0+)=0$ for $0\le j\le p$, extending $f$ by zero to $\theta<\theta_0$ introduces no boundary terms in its weak derivatives through order $p+1$. Together with the preceding $L^2$ estimates, this gives $f\in H^{p+1}$. With
\[
 \widehat f(s)=\frac1{\sqrt{2\pi}}\int_{\mathbb R} e^{-\ii s\theta}f(\theta)\,\dd\theta,
\]
Plancherel's theorem and the Cauchy--Schwarz inequality give
\begin{equation}\label{eq:Fourier_bound}
 \int_{\mathbb R}(1+s^2)^{p+1}|\widehat f(s)|^2\,\dd s\le K_p, \qquad \|\widehat f\|_{L^1}\le K_p.
\end{equation}
Set
\[
 h_1=\frac{\tau^p}{(\tau-\ii)^{p+2}}\frac1{\sqrt{2\pi}}\int_t^\infty e^{\ii s\theta}\widehat f(s)\,\dd s,\qquad h_2=\frac{\tau^p}{(\tau-\ii)^{p+2}}\frac1{\sqrt{2\pi}} \int_{-\infty}^t e^{\ii s\theta}\widehat f(s)\,\dd s.
\]
On the real axis, the Cauchy--Schwarz inequality and \eqref{eq:Fourier_bound} give $|h_1|\le K(1+\tau^2)^{-1}t^{-p-1/2}$. In the open sector $-\pi/4<\arg \tau<0$, $\re(\ii\theta)>0$, so the formula for $h_2$ defines an analytic function of $k$. After multiplication by $e^{-2\ii t\theta}$, the integrand is bounded by $e^{-t\re(\ii\theta)}|\widehat f(s)|$ for $s\le t$. On $\arg \tau=-\pi/4$,
\[
 |e^{-2\ii t\theta}h_2| \le K\frac{|\tau|^p e^{-2t|\tau|^2}}{|\tau-\ii|^{p+2}} \le K(1+|\tau|^2)^{-1}t^{-p/2}.
\]
The construction on the left half-line uses $3\pi/4<\arg \tau<\pi$, where $\re(\ii\theta)>0$ as well. The pole $\tau=\ii$ lies outside both closed sectors. Define the reflected columns in the conjugate sectors by
\[
 R_{\mathrm{low}}(k)=R(k^*)^\dagger,\qquad h_{2,\mathrm{low}}(k)=h_2(k^*)^\dagger.
\]
The subscript $\mathrm{low}$ refers to the lower triangular factor. Both reflected columns are analytic in $k$ in the conjugate sectors. The rational pole $\tau=-\ii$ lies outside $0\le\arg \tau\le\pi/4$ and $-\pi\le\arg \tau\le-3\pi/4$. On the real axis, the boundary values of $R_{\mathrm{low}}$ and $h_{2,\mathrm{low}}$ are $R^\dagger$ and $h_2^\dagger$, respectively. Reflection gives the corresponding estimates with $e^{2\ii t\theta}$. Shifted Schwartz seminorms bound the Taylor coefficients uniformly for $k_0$ in compact subsets of $\mathbb R$, and the choices $p\ge2l+2$, $m=3p+1$ give the prescribed order.\qed
\end{proof}

The two half-line constructions have endpoint values
\[
 R_{\mathrm{low}}(k_0)=R(k_0)^\dagger =\begin{cases}-\gamma(k_0)^\dagger/(1+|\gamma(k_0)|^2),&\text{left},\\
                \gamma(k_0)^\dagger,&\text{right}.
   \end{cases}
\]

\subsection{The first contour deformation}

Let $\Sigma=\mathbb R\cup L\cup L^*$, oriented as in Fig.~\ref{fig2:contour}, which shows the two central sectors $\Omega_1,\Omega_2$ and the four sectors adjacent to the real axis. The exponential factor $e^{-2\ii t\theta}$ in the upper off-diagonal block decays in $\Omega_5\cup\Omega_6$, while $e^{2\ii t\theta}$ in the lower off-diagonal block decays in $\Omega_3\cup\Omega_4$. Write $J^\Delta=b_-^{-1}b_+$, where
\[
 b_+=\mathsf U(e^{-2\ii t\theta}d\rho\delta),\qquad b_-=\mathsf L(-e^{2\ii t\theta}\delta^{-1}\rho^\dagger/d).
\]
On the real contour, the boundary values are those used in \eqref{eq:M_Delta_rh}. Write $b_\pm=b_\pm^o b_\pm^a$. The factors $b_+^o$ and $b_-^o$ contain the $h_1$ and $h_1^\dagger$ terms, respectively. The analytic factors in the corresponding sectors are
\[
 b_+^a=\mathsf U\bigl(e^{-2\ii t\theta}d(R+h_2)\delta\bigr),\qquad b_-^a=\mathsf L\bigl(-e^{2\ii t\theta} \delta^{-1}(R_{\mathrm{low}}+h_{2,\mathrm{low}})/d\bigr).
\]

\begin{figure}[t]
 \centering\includegraphics[width=.95\linewidth]{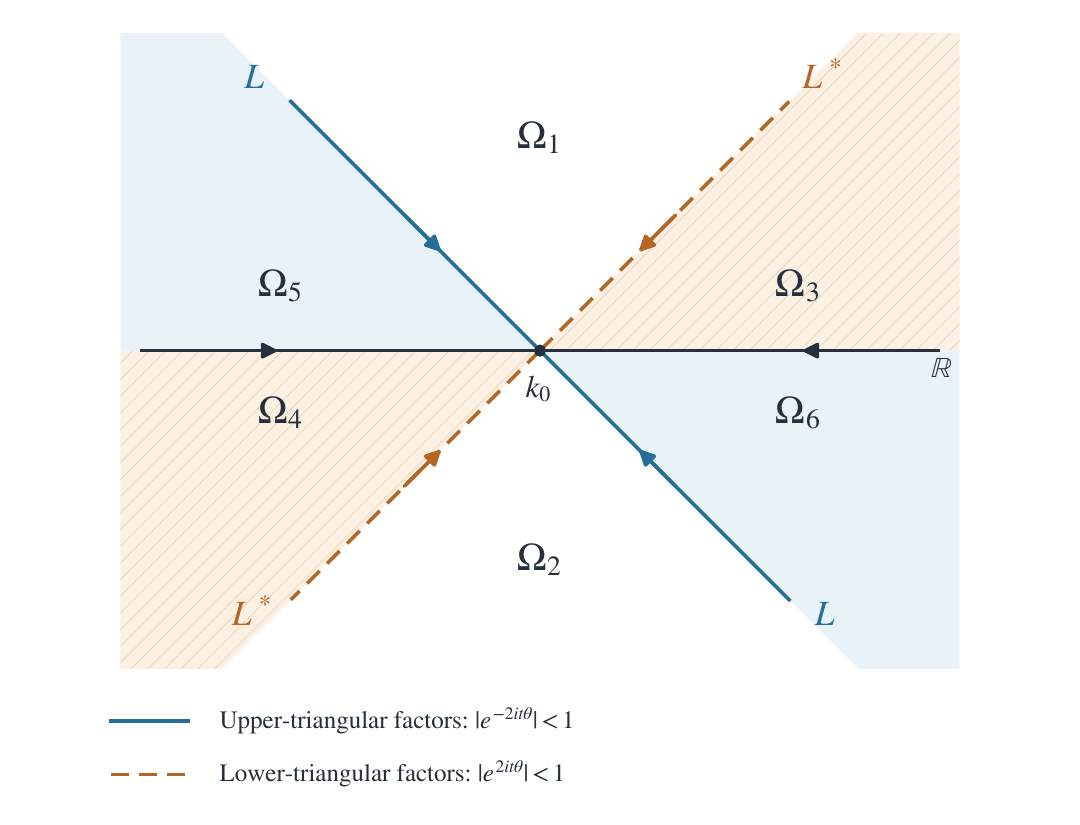}
 \caption{The first deformed contour $\Sigma=\mathbb R\cup L\cup L^*$ and its six sectors. All contour segments are oriented toward $k_0$. The solid oblique rays form $L$ and the dashed oblique rays form $L^*$. The unhatched shaded sectors satisfy $|e^{-2\ii t\theta}|<1$; the hatched sectors satisfy $|e^{2\ii t\theta}|<1$ for $t>0$.}
 \label{fig2:contour}
\end{figure}

\begin{lemma}\label{lem:deformation}
The piecewise analytic transformation
\[
 M^\sharp=\begin{cases}
 M^\Delta,&\Omega_1\cup\Omega_2,\\
 M^\Delta(b_-^a)^{-1},&\Omega_3\cup\Omega_4,\\
 M^\Delta(b_+^a)^{-1},&\Omega_5\cup\Omega_6
 \end{cases}
\]
has jump matrices $(b_-^o)^{-1}b_+^o$ on the real contour, $b_+^a$ on $L$, and $(b_-^a)^{-1}$ on $L^*$. As $k\to\infty$ in closed angular sectors separated from $\Sigma$,
\[
 M^\sharp(k)=\Id_3+\frac{M_1-\Delta_1}{k}+O(k^{-2}).
\]
\end{lemma}

\begin{proof}
Lemmas~\ref{lem:delta} and \ref{rational-de} imply that $M^\sharp$ is analytic in each sector. Using the sectorwise definitions of $M^\sharp$ and the jump matrix \eqref{eq:M_Delta_rh}, we obtain the stated jumps with the inward orientations shown in Fig.~\ref{fig2:contour}. As $k\to\infty$ in $\Omega_5\cup\Omega_6$, $e^{-2\ii t\theta}R=O(k^{-5})$ and $e^{-2\ii t\theta}h_2=O(k^{-2})$. Both estimates hold up to the real boundary of each sector. Reflection gives the corresponding estimates in $\Omega_3\cup\Omega_4$. The bounded quantities $d^{\pm1}$ and $\delta^{\pm1}$ therefore give $b_\pm^a=\Id_3+O(k^{-2})$ in their respective sectors. Together with $M^\Delta=M\Delta^{-1}$ and the expansion of $\Delta$ after \eqref{eq:Delta_rho}, this proves the stated expansion.\qed
\end{proof}

Write $J^\sharp=(\Id_3-\omega_-^\sharp)^{-1}(\Id_3+\omega_+^\sharp)$ and $\omega^\sharp=\omega_-^\sharp+\omega_+^\sharp$. We use the same sum convention for $\omega'$ and $\omega^0$. For each pair, we use the Beals--Coifman operator $C_\omega$ and density $\mathfrak m$ \cite{Beals-Coifman}:
\begin{equation}\label{eq:BC}
 C_\omega f=C_+(f\omega_-)+C_-(f\omega_+),\qquad \mathfrak m=\Id_3+C_\omega\mathfrak m,\qquad \mathfrak m-\Id_3\in L^2(\Sigma).
\end{equation}
The operator $C_\omega$ depends on the ordered pair $(\omega_-,\omega_+)$; the Cauchy representation of the RH solution has density $\mathfrak m\omega$. The boundary operators $C_\pm$ associated with the Cauchy operator $\Cauchy f(k)=(2\pi\ii)^{-1}\int_\Sigma f(\xi)/(\xi-k)\,\dd\xi$ are bounded on $L^2$ for this fixed finite union of straight rays.

For all triangular weight pairs used below, $\omega_+$ and $\omega_-$ are supported in the upper and lower off-diagonal blocks, respectively. Hence, in the Frobenius norm,
\[
 |\omega_-|^2+|\omega_+|^2=|\omega|^2,\qquad\|\omega_\pm\|_{L^p}\le\|\omega\|_{L^p} \quad(1\le p\le\infty).
\]
The same statements hold for differences of these weight pairs. In particular,
\[
 \|C_\omega\| \le \bigl(\|C_+\|+\|C_-\|\bigr)\|\omega\|_{L^\infty},
\]
and the corresponding operator differences are controlled by the $L^\infty$ norm of the difference of the weight sums.

\begin{proposition}[\cite{Beals-Coifman}]\label{prop:BC_representation}
If $1-C_{\omega^\sharp}$ is invertible on $L^2(\Sigma)$, the normalized RH solution is represented by
\begin{equation}\label{eq:affine_mu}
 \mathfrak m^\sharp-\Id_3=(1-C_{\omega^\sharp})^{-1}C_{\omega^\sharp}\Id_3,\qquad M^\sharp=\Id_3+\Cauchy(\mathfrak m^\sharp\omega^\sharp).
\end{equation}
Its coefficient of $k^{-1}$ and the reconstruction formula are
\begin{equation}\label{eq:M_sharp_solution}
 M_1^\sharp=-\frac1{2\pi\ii}\int_\Sigma\mathfrak m^\sharp\omega^\sharp\,\dd\xi,\qquad \ww=2(M_1)_{12}=2(M_1^\sharp)_{12}.
\end{equation}
These representations also apply to the other pairs of jump weights, with the corresponding Beals--Coifman solution and coefficient of $k^{-1}$. The inverse operator acts on $C_{\omega^\sharp}\Id_3\in L^2$. The last equality in \eqref{eq:M_sharp_solution} uses Lemma~\ref{lem:deformation} and the block-diagonal coefficient of $\Delta$.
\end{proposition}

\subsection{Reduction to the parabolic cylinder model}

We separate the $h_1$ and $h_2$ contributions from the rational terms in the jump weights. Lemma~\ref{lem:comparison} controls the resulting change in the coefficient of $k^{-1}$ in the RH solution.

Define the error weights $\omega^e_\pm$ by retaining the $h_1$ and $h_1^\dagger$ terms on $\mathbb R$, the $h_2$ terms on $L$, and the $h_{2,\mathrm{low}}$ terms on $L^*$. Set $\omega'_\pm=\omega^\sharp_\pm-\omega^e_\pm$ and $\omega^e=\omega^e_-+\omega^e_+$. The remaining pair is supported on $\Sigma'=L\cup L^*$ and is obtained by retaining only $R$ and $R_{\mathrm{low}}$ in the off-diagonal blocks. On $L$, the only nonzero block of $\omega'_+$ is $e^{-2\ii t\theta}dR\delta$. On $L^*$, the only nonzero block of $\omega'_-$ is $e^{2\ii t\theta}\delta^{-1}R_{\mathrm{low}}/d$.

\begin{lemma}\label{lem3.3}
For the decomposition in Lemma~\ref{rational-de} with fixed $l\ge1$,
\begin{equation}\label{eq:weight_norms}
 \|\omega^e\|_{L^1}+\|\omega^e\|_{L^2}+\|\omega^e\|_{L^\infty}=O(t^{-l}),\quad \|\omega'\|_{L^1}=O(t^{-1/2}),\quad\|\omega'\|_{L^2}=O(t^{-1/4}),\quad\|\omega'\|_{L^\infty}=O(1).
\end{equation}
\end{lemma}

\begin{proof}
The error estimates follow from Lemma~\ref{rational-de} and the bounds in \eqref{eq:delta_sym_bounds}. On $L$, $|e^{-2\ii t\theta}dR\delta|\le K e^{-4t|k-k_0|^2}(1+|k-k_0|^5)^{-1}$. The bounds for $\delta^{-1}$ and $d^{-1}$ give the same estimate for the lower triangular jump weights on $L^*$. Rescaling by $\sqrt t$ yields \eqref{eq:weight_norms}.\qed
\end{proof}

\begin{lemma}\label{lem:comparison}
Assume that $1-C_{\omega^\sharp}$ and $1-C_{\omega'}$ are invertible on $L^2(\Sigma)$, with inverse norms bounded by $K$, uniformly in the parameters under consideration. Let $\mathfrak m^\sharp$ and $\mathfrak m'$ be the corresponding solutions of \eqref{eq:BC}, with $\mathfrak m^\sharp-\Id_3,\mathfrak m'-\Id_3\in L^2(\Sigma)$. Then
\begin{equation}\label{eq:comparison_identity}
\begin{split}
 &\mathfrak m^\sharp-\mathfrak m'=(1-C_{\omega^\sharp})^{-1}C_{\omega^e}\mathfrak m',\\
 &\mathfrak m^\sharp\omega^\sharp-\mathfrak m'\omega'=\mathfrak m'\omega^e+[(1-C_{\omega^\sharp})^{-1}C_{\omega^e}\mathfrak m']\omega^\sharp.
\end{split}
\end{equation}
Moreover,
\begin{equation}\label{eq:comparison_bound}
 \left|\int_\Sigma(\mathfrak m^\sharp\omega^\sharp-\mathfrak m'\omega')\,\dd\xi\right|\le K'\bigl(\|\omega^e\|_{L^1}+\|\omega'\|_{L^2}\|\omega^e\|_{L^2}+(\|\omega^e\|_{L^2}+\|\omega^e\|_{L^\infty}\|\omega'\|_{L^2}) \|\omega^\sharp\|_{L^2}\bigr).
\end{equation}
\end{lemma}

\begin{proof}
Subtracting the two equations in \eqref{eq:BC} gives
\[
 (1-C_{\omega^\sharp})(\mathfrak m^\sharp-\mathfrak m')=C_{\omega^e}\mathfrak m'.
\]
Together with $\omega^\sharp=\omega'+\omega^e$, this proves \eqref{eq:comparison_identity}. Proposition~\ref{prop:BC_representation} gives $\|\mathfrak m'-\Id_3\|_{L^2}\le K'\|\omega'\|_{L^2}$. Separating the constant matrix from $\mathfrak m'$ gives
\[
 \|C_{\omega^e}\mathfrak m'\|_{L^2}\le K' (\|\omega^e\|_{L^2}+\|\omega^e\|_{L^\infty}\|\mathfrak m'-\Id_3\|_{L^2}).
\]
In particular, $C_{\omega^e}\mathfrak m'\in L^2(\Sigma)$. Moreover, $\|\mathfrak m'\omega^e\|_{L^1}\le\|\omega^e\|_{L^1}+\|\mathfrak m'-\Id_3\|_{L^2}\|\omega^e\|_{L^2}$. H\"older's inequality applied to the second identity in \eqref{eq:comparison_identity} proves \eqref{eq:comparison_bound}.\qed
\end{proof}

\begin{lemma}\label{lem:endpoint}
With $A_0$ as in \eqref{eq:endpoint_data}, write $\tau=k-k_0$ and set
\[
 \mathcal D(k)=\exp[\ii A_0\log(k-k_0)],\qquad F(\tau)=\mathcal D(k_0+\tau)^{-1}\delta(k_0+\tau).
\]
For $\tau<0$, the regular factor $F$ satisfies the jump relation
\begin{equation}\label{eq:F_jump}
 F_+=J_F(\tau)F_-,\qquad J_F(\tau)=\Id_2+\mathcal D_+^{-1}[\gamma(k_0+\tau)^\dagger\gamma(k_0+\tau)-\gamma(k_0)^\dagger\gamma(k_0)]\mathcal D_-.
\end{equation}
As $\tau\to0$ along the four oblique rays or along the positive real axis, $F(\tau)$ tends to $E_0$, and
\begin{equation}\label{eq:matrix_local}
 \delta(k)=\mathcal D(k)\{E_0+O(|\tau|(1+|\log|\tau||))\},\qquad E_0^\dagger E_0=\Id_2,\quad \det E_0=e^{\chi_0}.
\end{equation}
The estimates in \eqref{eq:matrix_local} are uniform for $k_0$ in compact subsets of $\mathbb R$, including points where $\gamma(k_0)=0$.
\end{lemma}

\begin{proof}
Since $e^{-2\pi A_0}=\Id_2+\gamma(k_0)^\dagger\gamma(k_0)$, we have $\mathcal D_+=(\Id_2+\gamma(k_0)^\dagger\gamma(k_0))\mathcal D_-$. Substitution into $F_+=\mathcal D_+^{-1}\delta_+$ gives \eqref{eq:F_jump}. Since $A_0$ is Hermitian and the argument on the chosen branch is bounded, $\mathcal D^{\pm1}$ are bounded; Lemma~\ref{lem:delta} then implies that $F$ and its boundary values are bounded. Smoothness of $\gamma$ gives $J_F(\xi)-\Id_2=O(\xi)$. Since $F_-$ is bounded, $[J_F(\xi)-\Id_2]F_-(\xi)=O(|\xi|)$ for $\xi<0$.

Choose a fixed small $\epsilon>0$. Cauchy's formula in the slit disk, with the outer circle traversed counterclockwise, gives
\[
 F(\tau)=\frac1{2\pi\ii}\oint_{|\zeta|=\epsilon} \frac{F(\zeta)}{\zeta-\tau}\,\dd\zeta+\frac1{2\pi\ii}\int_{-\epsilon}^0\frac{[J_F(\xi)-\Id_2]F_-(\xi)}{\xi-\tau}\,\dd\xi.
\]
For fixed nonzero $\tau$, the inner circle of radius $\varrho$ contributes $O(\varrho/|\tau|)$ and hence vanishes, because $F$ is bounded. The right-hand side is well defined at $\tau=0$; denote its value by $E_0$. Subtracting this value from $F(\tau)$ introduces the kernel difference $\tau/[\xi(\xi-\tau)]$ in the integral over the cut. For $|\tau|<\epsilon/2$ along the stated directions, $|\xi-\tau|\ge c(|\xi|+|\tau|)$, hence
\[
 |F(\tau)-E_0|\le K|\tau|+K|\tau|\int_0^\epsilon\frac{\dd\xi}{\xi+|\tau|}\le K|\tau|(1+|\log|\tau||).
\]
On $\tau>0$, both $\mathcal D$ and $\delta$ are unitary, so $E_0$ is unitary and is the limit appearing in \eqref{eq:endpoint_data}. Since $\det \mathcal D=\tau^{\ii\nu}$, \eqref{eq:d_local} gives $\det E_0=e^{\chi_0}$. For $k_0$ in a fixed compact set, the bounds for $\delta$, $\mathcal D^{\pm1}$, and the shifted reflection coefficient are uniform. The same $\epsilon$ can therefore be used in the circle and jump-density estimates. The continuous extension of $A_0$ at $\gamma(k_0)=0$ preserves these bounds.\qed
\end{proof}

The identities $\gamma(k_0)\mathcal D=\tau^{\ii\nu}\gamma(k_0)$ and $\mathcal D^{-1}\gamma(k_0)^\dagger=\tau^{-\ii\nu}\gamma(k_0)^\dagger$, together with \eqref{eq:matrix_local}, give the following expansions on the rays $\arg(k-k_0)=-\pi/4$ and $\arg(k-k_0)=\pi/4$, respectively:
\begin{equation}\label{eq:effective_products}
\begin{aligned}
 &dR\delta=\tau^{2\ii\nu}e^{\chi_0}\gamma(k_0)E_0(k_0)+O(|\tau|(1+|\log|\tau||)),\\
 &\frac{\delta^{-1}R_{\mathrm{low}}}{d}=\tau^{-2\ii\nu}e^{-\chi_0}E_0(k_0)^\dagger\gamma(k_0)^\dagger+O(|\tau|(1+|\log|\tau||)).
\end{aligned}
\end{equation}
On the two rays with $\re(k-k_0)<0$, the leading terms in \eqref{eq:effective_products} acquire the additional factor $-1/(1+|\gamma(k_0)|^2)$.

Set $z=\sqrt{8t}(k-k_0)$ and
\begin{equation}\label{eq:eta}
 \eta=(8t)^{-\ii\nu/2}e^{2\ii k_0^2t+\chi_0/2},\qquad |\eta|=1.
\end{equation}
The change of variables maps $\Sigma'$ to the cross $\Sigma^0=\bigcup_{j=1}^4\Sigma_j^0$ shown in Fig.~\ref{fig3:contour}. The rays $\Sigma_1^0$, $\Sigma_2^0$, $\Sigma_3^0$, and $\Sigma_4^0$ have arguments $-\pi/4$, $\pi/4$, $3\pi/4$, and $-3\pi/4$, respectively, and are all oriented toward the origin. We seek a matrix $M^0(z)$ that is analytic in $\mathbb C\setminus\Sigma^0$, is normalized at infinity, and remains bounded together with its inverse as $z\to0$ within each sector. Its jumps are
\begin{equation}\label{eq:M0_rh}
\begin{array}{c|c}
 \text{ray}& J^0\\\hline
 \Sigma^0_1&\mathsf U(\eta^2z^{2\ii\nu}e^{-\ii z^2/2}\gamma(k_0)E_0(k_0))\\
 \Sigma^0_2&\mathsf L(\eta^{-2}z^{-2\ii\nu}e^{\ii z^2/2}E_0(k_0)^\dagger\gamma(k_0)^\dagger)\\
 \Sigma^0_3&\mathsf U(-\eta^2z^{2\ii\nu}e^{-\ii z^2/2}\gamma(k_0)E_0(k_0)/(1+|\gamma(k_0)|^2))\\
 \Sigma^0_4&\mathsf L(-\eta^{-2}z^{-2\ii\nu}e^{\ii z^2/2}E_0(k_0)^\dagger\gamma(k_0)^\dagger/(1+|\gamma(k_0)|^2))
\end{array}
\end{equation}
The powers of $z$ use the same principal-branch convention, $-\pi<\arg z<\pi$. For the model jumps, set $\omega_+^0=J^0-\Id_3$ and $\omega_-^0=0$ on $\Sigma_1^0\cup\Sigma_3^0$; set $\omega_-^0=J^0-\Id_3$ and $\omega_+^0=0$ on $\Sigma_2^0\cup\Sigma_4^0$. As before, $\omega^0=\omega_-^0+\omega_+^0$. We extend both weights by zero to $\mathbb R$, so that they are defined on $\Sigma^0\cup\mathbb R$.

\begin{figure}[htbp]
 \centering\includegraphics[width=.95\linewidth]{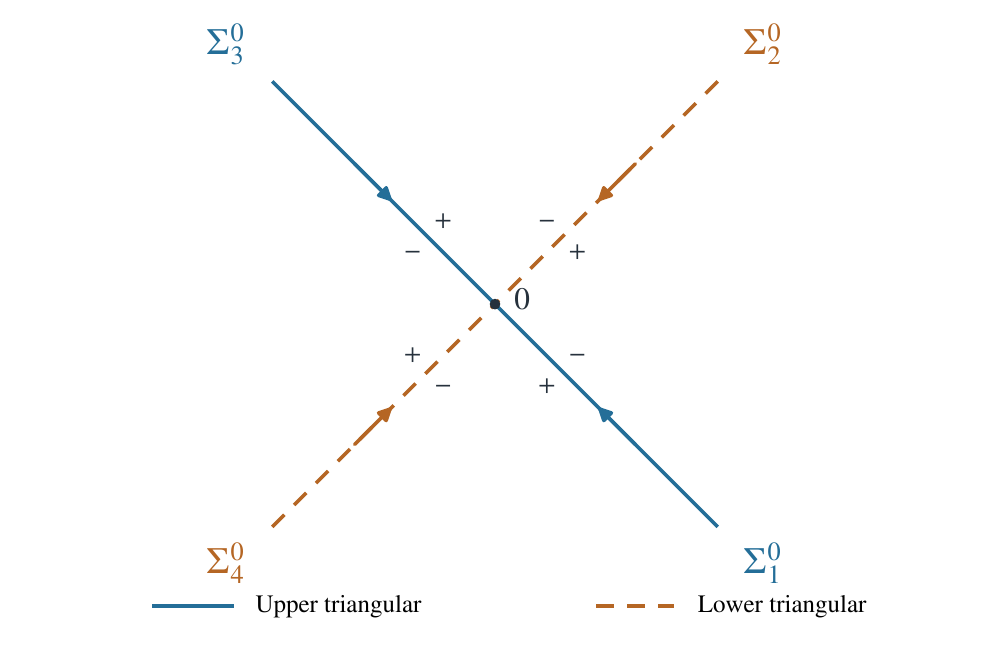}
  \caption{The four rays of the model cross $\Sigma^0$ in the $z$-plane are oriented toward the origin. Solid and dashed rays carry upper and lower triangular jumps, respectively. The plus and minus signs indicate the corresponding boundary sides.}
 \label{fig3:contour}
\end{figure}

\begin{proposition}\label{prop:model}
The model RH problem \eqref{eq:M0_rh} has a unique bounded matrix-valued solution $M^0(z)$. The operator $1-C_{\omega^0}$ is uniformly invertible on $L^2(\Sigma^0\cup\mathbb R)$ for $|\gamma(k_0)|\le\|\gamma\|_{L^\infty}$, unitary $E_0(k_0)$, and $|\eta|=1$. Writing $M^0(z)=\Id_3+M_1^0/z+O(z^{-2})$ as $z\to\infty$, we have $(M_1^0)_{11}=0$ and $(M_1^0)_{22}=0$, while
\begin{equation}\label{eq:model_coefficient}
 (M_1^0)_{12}=-\ii\eta^2\frac{e^{\pi\nu/2-\ii\pi/4}\Gamma(1+\ii\nu)}{\sqrt{2\pi}}\gamma(k_0)E_0(k_0).
\end{equation}
\end{proposition}

\begin{proof}
We first construct the model solution and compute its coefficient at infinity. We then establish uniform bounds for the model solution and for $(1-C_{\omega^0})^{-1}$.

For $\gamma(k_0)\ne0$, choose a unitary $V$ such that $\gamma(k_0)E_0(k_0)V=(|\gamma(k_0)|,0)$, and set $Q=\diag(1,V)$ and $\sigma_2=\diag(-1,1)$. Conjugation by $Q$ reduces the model to a $2\times2$ block and a scalar identity block. Diagonal conjugation removes the phase factors $\eta$ and gives a parabolic cylinder model, as in \cite[Section~3.6]{CNLS-DeiftZhou}:
\[
 Q^{-1}M^0Q=\diag(\eta^{-\sigma_2}H_2\eta^{\sigma_2},1).
\]

Put
\[
 \Psi=H_2z^{-\ii\nu\sigma_2}e^{\ii z^2\sigma_2/4},\qquad \beta=\frac{e^{\pi\nu/2-\ii\pi/4}\Gamma(1+\ii\nu)}{\sqrt{2\pi}}|\gamma(k_0)|.
\]
The identities $|\Gamma(1+\ii\nu)|^2=\pi\nu/\sinh(\pi\nu)$ and $1+|\gamma(k_0)|^2=e^{-2\pi\nu}$ give $|\beta|^2=-\nu$. Define the entire matrix function
\[
 \Psi_R(z)=e^{\pi\nu/4}
 \begin{pmatrix}
 D_{\ii\nu}(e^{\ii\pi/4}z)&\beta e^{-3\ii\pi/4}D_{-\ii\nu-1}(e^{-\ii\pi/4}z)\\
 -\beta^* e^{3\ii\pi/4}D_{\ii\nu-1}(e^{\ii\pi/4}z)&D_{-\ii\nu}(e^{-\ii\pi/4}z)
 \end{pmatrix}.
\]
Here, $D_\alpha$ denotes the parabolic cylinder function \cite{DLMF}. Set
\[
 \mathsf U_2(c)=\begin{pmatrix}1&c\\0&1\end{pmatrix},\qquad \mathsf L_2(c)=\begin{pmatrix}1&0\\c&1\end{pmatrix}.
\]
Define $\Psi$ sectorwise by
\[
\Psi(z)=
\begin{cases}
 \Psi_R(z),&-\pi/4<\arg z<\pi/4,\\
 \Psi_R(z)\mathsf L_2(-|\gamma(k_0)|),&\pi/4<\arg z<3\pi/4,\\
 \Psi_R(z)\mathsf L_2(-|\gamma(k_0)|)\mathsf U_2(|\gamma(k_0)|/(1+|\gamma(k_0)|^2)),&3\pi/4<\arg z<\pi,\\
 \Psi_R(z)\mathsf U_2(|\gamma(k_0)|)\mathsf L_2(-|\gamma(k_0)|/(1+|\gamma(k_0)|^2)),&-\pi<\arg z<-3\pi/4,\\
 \Psi_R(z)\mathsf U_2(|\gamma(k_0)|),&-3\pi/4<\arg z<-\pi/4.
\end{cases}
\]
With the inward orientations, the constant right jumps on $\Sigma_1^0$, $\Sigma_2^0$, $\Sigma_3^0$, and $\Sigma_4^0$ are $\mathsf U_2(|\gamma(k_0)|)$, $\mathsf L_2(|\gamma(k_0)|)$, $\mathsf U_2(-|\gamma(k_0)|/(1+|\gamma(k_0)|^2))$, and $\mathsf L_2(-|\gamma(k_0)|/(1+|\gamma(k_0)|^2))$, respectively. On $\Sigma_1^0$, the plus side lies below the ray, and the jump relation for the upper-right entry is
\[
 \beta e^{\ii\pi/4-3\pi\nu/4} D_{-\ii\nu-1}(e^{3\ii\pi/4}z)=|\gamma(k_0)|e^{\pi\nu/4}D_{\ii\nu}(e^{\ii\pi/4}z)+\beta e^{-3\ii\pi/4+\pi\nu/4}D_{-\ii\nu-1}(e^{-\ii\pi/4}z).
\]
The connection formula DLMF 12.2.18, written in $D_\alpha$ notation using DLMF 12.2.5 \cite{DLMF}, gives this jump relation when
\[
 \beta e^{\ii\pi/4-3\pi\nu/4}=\frac{|\gamma(k_0)|\Gamma(1+\ii\nu)}{\sqrt{2\pi}}e^{-\pi\nu/4}.
\]
This agrees with the stated value of $\beta$. The representations in the remaining sectors follow from the same connection formula at adjacent orders and from DLMF 12.2.19 at orders $\pm\ii\nu$. The triangular factors also satisfy
\[
 \begin{aligned}
 \mathsf L_2(-|\gamma(k_0)|)\mathsf U_2\!\left(\frac{|\gamma(k_0)|}{1+|\gamma(k_0)|^2}\right)
 &=\mathsf U_2(|\gamma(k_0)|)\mathsf L_2\!\left(-\frac{|\gamma(k_0)|}{1+|\gamma(k_0)|^2}\right)\\
 &\quad\times\diag(1+|\gamma(k_0)|^2,(1+|\gamma(k_0)|^2)^{-1}).
 \end{aligned}
\]
Thus, the boundary values across the negative real axis satisfy
\[
 \Psi_+(z)=\Psi_-(z)\diag(1+|\gamma(k_0)|^2,(1+|\gamma(k_0)|^2)^{-1}),\qquad z<0,
\]
where $+$ and $-$ here refer to the upper and lower sides. Since $1+|\gamma(k_0)|^2=e^{-2\pi\nu}$, this jump is canceled by the principal-branch power in $H_2=\Psi e^{-\ii z^2\sigma_2/4}z^{\ii\nu\sigma_2}$. Hence, $H_2$ is analytic across the negative real axis.

Using the connection formulas, we write the first column of $\Psi_R \mathsf L_2(-|\gamma(k_0)|)$ with argument $e^{-3\ii\pi/4}z$ and the second column of $\Psi_R \mathsf U_2(|\gamma(k_0)|)$ with argument $e^{3\ii\pi/4}z$. Together with the original arguments $e^{\pm\ii\pi/4}z$, these representations cover all five sectors. In each representation, $\zeta$ denotes the argument of the parabolic cylinder function and satisfies $|\arg\zeta|\le\pi/2$, strictly inside the asymptotic sector $|\arg\zeta|<3\pi/4$. The orders of the parabolic cylinder functions are $\pm\ii\nu$ and $\pm\ii\nu-1$, where
\[
 -\frac{\log(1+\|\gamma\|_{L^\infty}^2)}{2\pi}\le \nu\le0.
\]
For these orders and $|\arg\zeta|\le\pi/2$, DLMF 12.9.1 and the remainder bounds in \S12.9(ii), obtained from \S13.7(ii), apply uniformly. Substitution into the sector formulas yields
\[
 H_2=\Id_2+\frac1z
   \begin{pmatrix}0&-\ii\beta\\-\ii\beta^*&0\end{pmatrix}
       +O(z^{-2}).
\]
Conjugation by the diagonal factors $\eta^{\mp\sigma_2}$ and then by $Q$ gives \eqref{eq:model_coefficient} and the vanishing diagonal blocks.

The functions $D_\alpha(\zeta)$ are entire in $\zeta$, and the imaginary powers are bounded in each sector. Hence, $H_2$ is bounded near the origin. Since all jump matrices have determinant one, $\det H_2$ has no jumps and its singularity at the origin is removable. The normalization at infinity and Liouville's theorem therefore give $\det H_2=1$. The uniform expansion at infinity, continuity of $D_\alpha(\zeta)$ on compact sets of $(\alpha,\zeta)$, and bounded imaginary powers give a uniform bound for $H_2$ when $|\gamma(k_0)|\le\|\gamma\|_{L^\infty}$. The identity $\det H_2=1$ gives the corresponding bound for $H_2^{-1}$. Unitarity of $V$ and $E_0$, together with $|\eta|=1$, preserves these bounds for $M^0$ and $(M^0)^{-1}$. If $\gamma(k_0)=0$, then $M^0\equiv\Id_3$.

For any other normalized solution $\widetilde M^0$, the product $\widetilde M^0(M^0)^{-1}$ has no jumps, is bounded near the origin, and tends to $\Id_3$ at infinity. The singularity at the origin is removable, so Liouville's theorem gives $\widetilde M^0(M^0)^{-1}=\Id_3$, proving uniqueness. Boundedness near the origin and the expansion at infinity give $M^0_\pm-\Id_3\in L^2(\Sigma^0\cup\mathbb R)$.

Set $b_-^0=\Id_3-\omega_-^0$ and $b_+^0=\Id_3+\omega_+^0$, so that $J^0=(b_-^0)^{-1}b_+^0$. For injectivity, suppose that $h\in L^2(\Sigma^0\cup\mathbb R)$ solves $(1-C_{\omega^0})h=0$, and set
\[
 \mathcal Y(z)=\Cauchy(h\omega^0)(z),\qquad \Xi(z)=\mathcal Y(z)(M^0(z))^{-1}.
\]
Since $\mathcal Y_-=h b_-^0$ and $\mathcal Y_+=h(\Id_3+\omega_+^0)$, $\mathcal Y_+=\mathcal Y_-J^0$. Thus, $\Xi$ has no jumps and extends analytically across every open ray. Also, $h\omega^0\in L^1\cap L^2$. Define
\[
 \mathcal K_{\varrho}(s)=\int_{|z|=\varrho}\frac{|\dd z|}{|s-z|}.
\]
The Cauchy representation and the bounded model inverse give
\[
 \int_{|z|=\varrho}|\Xi(z)|\,|\dd z|\le K\|h\omega^0\|_{L^2}\|\mathcal K_{\varrho}\|_{L^2(\Sigma^0\cup\mathbb R)}\le K'\varrho^{1/2}\|h\omega^0\|_{L^2}.
\]
Indeed, the change of variables $s=\varrho\zeta$ gives $\mathcal K_{\varrho}(\varrho\zeta)=\mathcal K_1(\zeta)$. On each ray, $\mathcal K_1$ has only logarithmic singularities at its intersections with the unit circle and is $O(|\zeta|^{-1})$ at infinity, so it belongs to $L^2$. For every $n\ge1$, the coefficient $\Xi_{-n}$ of $z^{-n}$ in the Laurent expansion of $\Xi$ satisfies
\[
 |\Xi_{-n}|\le \frac{\varrho^{n-1}}{2\pi}\int_{|z|=\varrho}|\Xi(z)|\,|\dd z|\le K'\varrho^{n-1/2}\|h\omega^0\|_{L^2}\longrightarrow0\qquad(\varrho\downarrow0).
\]
Hence, the origin is removable. The same circle estimate holds for every radius $R>0$. For any fixed $z_0$ and $R>2|z_0|$, Cauchy's formula therefore yields
\[
 |\Xi(z_0)|\le K'R^{-1/2}\|h\omega^0\|_{L^2}\longrightarrow0\qquad(R\to\infty).
\]
Thus $\Xi=0$, and $\mathcal Y_-=h b_-^0$ with invertible $b_-^0$ gives $h=0$.

To prove surjectivity, fix $f\in L^2(\Sigma^0\cup\mathbb R)$ and define
\[
 g=f(J^0-\Id_3)(M_+^0)^{-1},\qquad \mathcal X=\Cauchy g,\qquad \mathcal Y=\mathcal XM^0,\qquad \mathfrak m=(f+\mathcal Y_-)(b_-^0)^{-1}.
\]
The bounded model multipliers and the $L^2$ bounds for the Cauchy boundary operators give $\mathfrak m\in L^2$. The inclusions $J^0-\Id_3,\omega^0\in L^2\cap L^\infty$ then imply $g,\mathfrak m\omega^0\in L^1\cap L^2$. The Plemelj formulas and $M_+^0=M_-^0J^0$ yield
\[
 \mathcal Y_+-\mathcal Y_-=(f+\mathcal Y_-)(J^0-\Id_3)=\mathfrak m b_-^0(J^0-\Id_3)=\mathfrak m\omega^0.
\]
Thus, $Z=\mathcal Y-\Cauchy(\mathfrak m\omega^0)$ has no jumps and extends analytically across every open ray. The preceding Cauchy-kernel estimate, now using the bounded multiplier $M^0$, gives, for every $R>0$,
\[
 \int_{|z|=R}|Z(z)|\,|\dd z|\le K R^{1/2}\bigl(\|g\|_{L^2}+\|\mathfrak m\omega^0\|_{L^2}\bigr).
\]
Letting $R\downarrow0$ in the circle estimate shows that all negative Laurent coefficients of $Z$ vanish, so the origin is removable. Cauchy's formula on circles with $R\to\infty$ then gives $Z\equiv0$. Consequently,
\[
 \mathfrak m b_-^0=f+C_-(\mathfrak m\omega^0).
\]
Using $C_+-C_-=1$, this identity is equivalent to $(1-C_{\omega^0})\mathfrak m=f$. Together with injectivity, this proves invertibility and gives
\[
 (1-C_{\omega^0})^{-1}f = \left[ f+C_-\!\left(f(J^0-\Id_3)(M_+^0)^{-1}\right)M_-^0\right](b_-^0)^{-1}.
\]

The uniform bounds for the model boundary values, their inverses, and the triangular jump factors, together with the $L^2$ operator bounds for $C_\pm$, yield the asserted inverse bound on $L^2(\Sigma^0\cup\mathbb R)$. The normalized Beals--Coifman equation is solved in $\Id_3+L^2(\Sigma^0\cup\mathbb R)$, or equivalently by solving the $L^2$ equation for $\mathfrak m-\Id_3$ in \eqref{eq:affine_mu}.\qed
\end{proof}

\subsection{Comparison with the parabolic cylinder model}

\begin{lemma}\label{lem:scaled_error}
Let $\widehat\omega_\pm(z)=\omega'_\pm(k_0+z/\sqrt{8t})$ and $\widehat\omega=\widehat\omega_-+\widehat\omega_+$. Then
\begin{equation}\label{eq:scaled_error}
 \|\widehat\omega-\omega^0\|_{L^1}+ \|\widehat\omega-\omega^0\|_{L^2}+ \|\widehat\omega-\omega^0\|_{L^\infty}=O(t^{-1/2}\log t).
\end{equation}
\end{lemma}

\begin{proof}
For $|\tau|<\epsilon$, the rational factor differs from its half-line endpoint value by $O(|\tau|)$. The scalar and matrix endpoint expansions \eqref{eq:d_local} and \eqref{eq:matrix_local} therefore give the product estimates in \eqref{eq:effective_products}, with errors of order $|\tau|(1+|\log|\tau||)$. The bounds for $d^{\pm1}$, $\delta^{\pm1}$, and $\mathcal D^{\pm1}$ also control the remaining factors, including $F^{\pm1}$. For the inverse matrix factor, we use
\[
 F^{-1}-E_0^{-1}=F^{-1}(E_0-F)E_0^{-1},
\]
which gives the same error bound as for $F-E_0$, since the inverse factors are bounded. After scaling, each error is bounded on its ray by
\[
 K t^{-1/2}|z|(1+\log t+|\log|z||)e^{-c|z|^2}.
\]
The right-hand side tends to zero as $z\to0$ and belongs to $L^1\cap L^2\cap L^\infty$. Outside the disk $|k-k_0|<\epsilon$, the scaled weights satisfy
\[
 |\widehat\omega(z)|+|\omega^0(z)|\le K e^{-|z|^2/2},\qquad |z|\ge\sqrt{8t}\,\epsilon,
\]
on the four rays. Consequently, their tails have $L^1$, $L^2$, and $L^\infty$ norms of order $O(e^{-c_\epsilon t})$ for some $c_\epsilon>0$. Finally, the identities $e^{-2\ii t\theta}=e^{4\ii k_0^2t}e^{-\ii z^2/2}$ and $\tau^{2\ii\nu}=(8t)^{-\ii\nu}z^{2\ii\nu}$ yield \eqref{eq:eta} and \eqref{eq:M0_rh}. By Lemma~\ref{lem:endpoint}, these bounds remain uniform at points $k_0$ where $\gamma(k_0)=0$.\qed
\end{proof}

Consider the RH problem for a matrix $M'$ analytic in $\mathbb C\setminus\Sigma'$, with jump relation
\[
 M'_+=M'_-(\Id_3-\omega'_-)^{-1}(\Id_3+\omega'_+)
\]
and normalization $M'(\infty)=\Id_3$. We use the Beals--Coifman class: $M'=\Id_3+\Cauchy(\mathfrak m'\omega')$, where $\mathfrak m'-\Id_3\in L^2(\Sigma)$ and $\mathfrak m'=\Id_3+C_{\omega'}\mathfrak m'$, with both weights extended by zero to the real line. The inverse bound proved below gives existence and uniqueness in this class for sufficiently large $t$. For its solution, we write
\[
 M'(k)=\Id_3+\frac{M'_1}{k}+O(k^{-2})
\]
as $k\to\infty$ in closed angular sectors separated from $\Sigma'$.

\begin{proposition}\label{eq:main_reduction}
For sufficiently large $t$, the inverses of $1-C_{\omega^0}$ and $1-C_{\widehat\omega}$ on $L^2(\Sigma^0\cup\mathbb R)$ are uniformly bounded, as are the inverses of $1-C_{\omega'}$ and $1-C_{\omega^\sharp}$ on $L^2(\Sigma)$. Moreover,
\begin{equation}\label{eq:full_coefficient}
 M_1=\Delta_1+\frac{M_1^0}{\sqrt{8t}}+O_C(t^{-1}\log t),\qquad \ww=\frac{(M_1^0)_{12}}{\sqrt{2t}}+O_C(t^{-1}\log t).
\end{equation}
Here, $M^0=\Id_3+M_1^0/z+O(z^{-2})$. The first estimate holds for the full coefficient matrix in the Frobenius norm. Both estimates are uniform for $|x/t|\le C$.
\end{proposition}

\begin{proof}
Proposition~\ref{prop:model} gives the inverse on the scaled cross, including zero weights on the real line. By Lemma~\ref{lem:scaled_error},
\[
 \|(1-C_{\omega^0})^{-1}(C_{\widehat\omega}-C_{\omega^0})\|\le K\|\widehat\omega-\omega^0\|_{L^\infty} \le K t^{-1/2}\log t<\tfrac12
\]
for large $t$. A Neumann series applied to the preceding perturbation yields a uniform bound for $(1-C_{\widehat\omega})^{-1}$. The map $f(k)\mapsto(8t)^{-1/4}f(k_0+z/\sqrt{8t})$ is unitary from $L^2(\Sigma)$ to $L^2(\Sigma^0\cup\mathbb R)$ and intertwines $C_{\omega'}$ with $C_{\widehat\omega}$. Hence, the same inverse bound holds for $1-C_{\omega'}$. By Lemma~\ref{lem3.3}, $\|C_{\omega^e}\|\le K\|\omega^e\|_{L^\infty}=O(t^{-l})$; a second perturbation gives the inverse for $1-C_{\omega^\sharp}$. Proposition~\ref{prop:BC_representation} constructs $M'$ with $\omega'$ extended by zero to $\Sigma$.

The coefficient formula \eqref{eq:M_sharp_solution} and the change of variables $z=\sqrt{8t}(k-k_0)$ give
\[
 M'_1=-\frac1{2\pi\ii\sqrt{8t}}\int_{\Sigma^0}\widehat{\mathfrak m}(z)\widehat\omega(z)\,\dd z,\qquad M_1^0=-\frac1{2\pi\ii}\int_{\Sigma^0}\mathfrak m^0(z)\omega^0(z)\,\dd z,
\]
where $\mathfrak m^0$ solves the Beals--Coifman equation \eqref{eq:BC} associated with $\omega^0$, and $\widehat{\mathfrak m}(z)=\mathfrak m'(k_0+z/\sqrt{8t})$ solves the corresponding equation for the rescaled weights $\widehat\omega_\pm$ on $\Sigma^0\cup\mathbb R$. The $L^1$, $L^2$, and $L^\infty$ norms of $\omega^0$, as well as $\|\mathfrak m^0-\Id_3\|_{L^2}$, are uniformly bounded. Applying Lemma~\ref{lem:comparison} to the weight pairs $\widehat\omega_\pm$ and $\omega^0_\pm$, and using \eqref{eq:scaled_error}, gives an $O(t^{-1/2}\log t)$ bound for the difference of the matrix-valued integrals. The prefactor $1/\sqrt{8t}$ reduces this error to $O(t^{-1}\log t)$. On $\Sigma$, \eqref{eq:comparison_bound} and \eqref{eq:weight_norms} give $M_1^\sharp-M'_1=O(t^{-l})$; choose $l\ge2$. Lemma~\ref{lem:deformation} gives $M_1^\sharp=M_1-\Delta_1$. This identity and the preceding estimates prove \eqref{eq:full_coefficient}.\qed
\end{proof}

\section{Proof of Theorem~\ref{thm1}}\label{sec:5}

Combining the reconstruction formula $\ww=2(M_1)_{12}$ with Proposition~\ref{eq:main_reduction}, \eqref{eq:model_coefficient}, and \eqref{eq:eta} immediately yields the long-time asymptotic expansion for the CNLS solution:
\begin{equation}\label{eq:CNLS_asymptotics}
 \ww(x,t)=\frac{-\ii\Gamma(1+\ii\nu)e^{\pi\nu/2+4\ii k_0^2t+\chi_0-\ii\pi/4-\ii\nu\log(8t)}}{2\sqrt{\pi t}}\gamma(k_0)E_0(k_0)+O_C(t^{-1}\log t).
\end{equation}
To recover the CKE solution through \eqref{eq:gauge}, we estimate the gauge phase using the $(1,1)$ entry of $M_1$ and the trace formula for the conserved mass.

\begin{proposition}\label{prop:mass}
Let $S=|\ww|^2$. For each fixed $t\ge0$,
\begin{equation}\label{eq:tail_identity}
\begin{aligned}
 \int_x^\infty S(y,t)\,\dd y&=-2\ii(M_1)_{11}(x,t),\\
 \int_{\mathbb R}S(y,t)\,\dd y&=\mathcal M_1=\frac1\pi\int_{\mathbb R}\log(1+|\gamma(\xi)|^2)\,\dd\xi.
\end{aligned}
\end{equation}
Moreover, as $t\to+\infty$, uniformly for $k_0$ in compact subsets of $\mathbb R$,
\begin{equation}\label{eq:diagonal_asymptotics}
 (M_1)_{11}=d^{[1]}+O(t^{-1}\log t),\qquad d^{[1]}=\frac\ii{2\pi}\int_{-\infty}^{k_0}\log(1+|\gamma(\xi)|^2)\,\dd\xi.
\end{equation}
\end{proposition}

\begin{proof}
The Volterra normalization of $\mu_{+,L}$ at $x=+\infty$ fixes the integration constant in \eqref{eq:reconstruction} for each $t$. The first column of $M_-$ in \eqref{eq:Jost_RH} is $\mu_{+,L}$. Its $k^{-1}$ coefficient has lower $2\times1$ block $-\ww^\dagger/2$ and upper scalar entry
\[
 (M_1)_{11}(x,t) =-\int_x^\infty\ii\ww(y,t) [-\ww(y,t)^\dagger/2]\,\dd y=\frac\ii2\int_x^\infty S(y,t)\,\dd y.
\]
This proves the first identity in \eqref{eq:tail_identity}.

Since $a(k)$ is independent of $t$, formula \eqref{abk} applied to $\ww(\cdot,t)$ gives the large-$k$ expansions
\[
 a(k)=\Id_2+\frac\ii{2k}\int_{\mathbb R}\ww^\dagger\ww\,\dd x+O(k^{-2}), \qquad \det a(k)=1+\frac\ii{2k}\int_{\mathbb R}S\,\dd x+O(k^{-2}).
\]
The orthonormality of the last two columns of the scattering matrix gives $a^\dagger(\Id_2+\gamma^\dagger\gamma)a=\Id_2$, and hence $|\det a|^2=(1+|\gamma|^2)^{-1}$ on $\mathbb R$. Consider the scalar function defined as $\det a(k)$ for $k\in\mathbb C_-$ and as $\bigl[(\det a(k^*))^*\bigr]^{-1}$ for $k\in\mathbb C_+$. This function is analytic in $\mathbb C\setminus\mathbb R$, has bounded inverse by Assumption~\ref{ass:data}, and its upper boundary value equals $(1+|\gamma(k)|^2)$ times its lower boundary value. This scalar jump relation and the normalization to one at infinity give the representation
\[
 \det a(k)=\exp\left[\frac1{2\pi\ii} \int_{\mathbb R}\frac{\log(1+|\gamma(\xi)|^2)}{\xi-k}\,\dd\xi\right], \qquad k\in\mathbb C_-.
\]
Comparing its $k^{-1}$ coefficient with the preceding expansion proves the total-mass trace formula. Since $\det a(k)\ne0$ for all $k\in\mathbb C_-$ by Assumption~\ref{ass:data}, the scalar factorization involves no Blaschke factors.

Expanding \eqref{eq:det_delta} yields the stated $d^{[1]}$. Since the diagonal blocks of $M_1^0$ vanish by Proposition~\ref{prop:model}, the full matrix estimate \eqref{eq:full_coefficient} gives \eqref{eq:diagonal_asymptotics}.\qed
\end{proof}

Let $I(x,t)=\int_{-\infty}^xS(y,t)\,\dd y$. By Proposition~\ref{prop:mass}, subtracting the tail in \eqref{eq:tail_identity} from the total mass and using \eqref{eq:diagonal_asymptotics} gives
\[
 I(x,t)=\frac1\pi\int_{k_0}^\infty\log(1+|\gamma(\xi)|^2)\,\dd\xi+O_C(t^{-1}\log t).
\]
This proves \eqref{eq:mass_asymptotics}. The identity $\qq\qq^\dagger=\ww\ww^\dagger$ gives the same estimate for the cumulative mass of the CKE solution. For real $I$ and $I_0$,
\[
 \begin{aligned}
 &|\Theta(x,t)-\lambda I_0(k_0)|=|\lambda|\,|I(x,t)-I_0(k_0)|=O_C(|\lambda|t^{-1}\log t),\\
 &|e^{\ii\lambda I(x,t)}-e^{\ii\lambda I_0(k_0)}|\le |\lambda|\,|I(x,t)-I_0(k_0)|=O_C(|\lambda|t^{-1}\log t).
 \end{aligned}
\]

Let $\ww_{\mathrm{lead}}$ denote the leading term in \eqref{eq:CNLS_asymptotics}. By \eqref{eq:gauge},
\[
 \qq-\ww_{\mathrm{lead}}e^{\ii\lambda I_0}=(\ww-\ww_{\mathrm{lead}})e^{\ii\lambda I}+\ww_{\mathrm{lead}}(e^{\ii\lambda I}-e^{\ii\lambda I_0}).
\]
The first term on the right-hand side is $O_C(t^{-1}\log t)$ since $|e^{\ii\lambda I}|=1$. The second term is $O_C(|\lambda|t^{-3/2}\log t)$, since $\ww_{\mathrm{lead}}=O_C(t^{-1/2})$. This establishes \eqref{eq:CKE_asymptotics} and completes the proof of Theorem~\ref{thm1}.

\begin{remark}\label{rem:joint_control}
Although the leading term $\ww_{\mathrm{lead}}$ in \eqref{eq:CNLS_asymptotics} satisfies the exact identity $\int_{-\infty}^x|\ww_{\mathrm{lead}}(y,t)|^2\,\dd y=I_0(k_0)$, the pointwise estimate \eqref{eq:CNLS_asymptotics} alone does not yield the cumulative-mass estimate \eqref{eq:mass_asymptotics} by direct spatial integration. Indeed, for fixed $C>0$, the bounds $\ww_{\mathrm{lead}}=O_C(t^{-1/2})$ and $\ww-\ww_{\mathrm{lead}}=O_C(t^{-1}\log t)$ imply that, uniformly for $|y|\le Ct$,
\[
 \bigl||\ww|^2-|\ww_{\mathrm{lead}}|^2\bigr|\le 2|\ww_{\mathrm{lead}}|\,|\ww-\ww_{\mathrm{lead}}|+|\ww-\ww_{\mathrm{lead}}|^2=O_C(t^{-3/2}\log t).
\]
Splitting the integral at $y=-Ct$ and integrating this bound over $[-Ct,x]$ for $|x|\le Ct$ gives
\[
 I(x,t)-I_0(k_0)=\int_{-\infty}^{-Ct}\bigl(|\ww(y,t)|^2-|\ww_{\mathrm{lead}}(y,t)|^2\bigr)\,\dd y+O_C(t^{-1/2}\log t).
\]
On the one hand, the resulting $O_C(t^{-1/2}\log t)$ bound for the integral over $[-Ct,x]$ is weaker than the $O_C(t^{-1}\log t)$ remainder in \eqref{eq:mass_asymptotics}. On the other hand, the pointwise estimate on $|y|\le Ct$ does not control the remaining tail integral over $(-\infty,-Ct)$, and Schwartz decay of $\ww(\cdot,t)$ at each fixed time does not by itself provide the required estimate for this tail difference as $t\to\infty$. In fact, the leading tail mass is
\[
 \int_{-\infty}^{-Ct}|\ww_{\mathrm{lead}}(y,t)|^2\,\dd y=\frac1\pi\int_{C/4}^{\infty}\log(1+|\gamma(\xi)|^2)\,\dd\xi,
\]
which is independent of $t$ and generally nonzero for fixed $C$. Thus, the tail region $y<-Ct$ cannot simply be omitted for fixed $C$, whereas allowing $C$ to grow with $t$ would require quantitative control of the dependence on $C$ in the remainder in \eqref{eq:CNLS_asymptotics}. Proposition~\ref{prop:mass} provides the required control through the exact relation $I(x,t)=\mathcal M_1+2\ii(M_1)_{11}(x,t)$. Since $(M_1^0)_{11}=0$, the full coefficient expansion \eqref{eq:full_coefficient} yields the $O_C(t^{-1}\log t)$ remainder for the cumulative mass and the gauge phase.
\end{remark}

\section{Numerical evaluation of the asymptotic solution}\label{sec:6}

We evaluate the leading term $\qq_{\mathrm{lead}}=(q_{1,\mathrm{lead}},q_{2,\mathrm{lead}})$ in \eqref{eq:CKE_asymptotics} from the reflection coefficient, the scalar integrals, and the unitary endpoint factor $E_0$.

We choose the CNLS initial data
\[
 \ww_0(x)=0.35\bigl(e^{-(x+0.6)^2+0.7\ii x},\,0.8e^{-(x-0.4)^2-0.5\ii x+\ii\pi/3}\bigr).
\]
For $\lambda=0.05$, the gauge transformation \eqref{eq:gauge} gives the CKE initial data
\[
 \qq_0(x)=\ww_0(x)\exp\!\left[\ii\lambda\int_{-\infty}^x|\ww_0(y)|^2\,\dd y\right].
\]
The phase integral is evaluated using the error function. To verify Assumption~\ref{ass:data}, we use the Volterra-series bound
\[
 |a(k)-\Id_2|\le\cosh\!\left(\int_{\mathbb R}|\ww_0(x)|\,\dd x\right)-1<1,\qquad k\in\overline{\mathbb C}_-,
\]
because $\int_{\mathbb R}|\ww_0(x)|\,\dd x\le0.35(1+0.8)\sqrt\pi<\operatorname{arcosh}2$. Thus $a(k)$ is invertible in the closed lower half-plane, and Assumption~\ref{ass:data} holds.

\subsection{Numerical direct scattering}

We compute the Jost solutions of the spatial Lax equation at $t=0$ by mapped Chebyshev collocation \cite[Sections~4.3.2 and~9.2]{TrogdonOlver2016}. The scattering matrix and reflection coefficient are recovered at $x=t=0$ from
\[
 s(k)=\mu_+(0,0;k)^{-1}\mu_-(0,0;k),\qquad a(k)=s_{22}(k),\quad b(k)=s_{12}(k),\quad\gamma(k)=b(k)a(k)^{-1}.
\]
For real $k$, the corrections $\phi_\pm(x;k)=\mu_\pm(x,0;k)-\Id_3$ satisfy
\[
 (\phi_\pm)_x-\ii k[\sigma,\phi_\pm]-U(x,0)\phi_\pm=U(x,0),\qquad \phi_\pm\to0\quad(x\to\pm\infty).
\]

Let $\xi_\ell=\cos[\ell\pi/(N-1)]$, $0\le\ell<N$, be the Chebyshev--Lobatto nodes. We map the two half-lines to $[-1,1]$ by
\[
 \mathcal H_-(x)=2\tanh(0.2x)+1\quad(x\le0),\qquad \mathcal H_+(x)=2\tanh(0.2x)-1\quad(x\ge0).
\]
Set $x_{\pm,\ell}=\mathcal H_\pm^{-1}(\xi_\ell)$, with the endpoint values understood as limits. For the Chebyshev polynomials $T_m$, define
\[
 \mathsf T_{\ell m}=T_m(\xi_\ell),\qquad\mathsf T'_{\ell m}=T'_m(\xi_\ell),\qquad 0\le\ell,m<N.
\]
The differentiation matrices with respect to $x$ are
\[
 \mathsf D_\pm=\diag\bigl(\mathcal H'_\pm(x_{\pm,\ell})\bigr)\mathsf T'\mathsf T^{-1}.
\]
Writing $\phi_{\pm,j}$ for column $j$ of $\phi_\pm$, collocation at the finite nodes gives
\[
 \sum_{m=0}^{N-1}(\mathsf D_\pm)_{\ell m}\phi_{\pm,j}(x_{\pm,m};k)-\bigl[\ii k(\sigma-\sigma_{jj}\Id_3)+U(x_{\pm,\ell},0)\bigr]\phi_{\pm,j}(x_{\pm,\ell};k)=U(x_{\pm,\ell},0)e_j,
\]
where $j=1,2,3$, $\sigma_{jj}$ is the $j$th diagonal entry of $\sigma$, and $e_j$ is the $j$th column of $\Id_3$. At the infinite endpoints, we replace the differential equation by $\phi_-=0$ at $x_{-,N-1}=-\infty$ and $\phi_+=0$ at $x_{+,0}=+\infty$. For each half-line and column, this gives a $3N\times3N$ linear system. Columns two and three share the same coefficient matrix because $\sigma_{22}=\sigma_{33}$.

The values at $x_{-,0}=x_{+,N-1}=0$ give $\mu_\pm(0,0;k)=\Id_3+\phi_\pm(0;k)$. We then solve $\mu_+s=\mu_-$ and $\gamma a=b$ for the scattering matrix and reflection coefficient.

For the spectral interpolation, we use $N=128$ spatial nodes on each half-line and evaluate $\gamma$ at the 321 Chebyshev--Lobatto nodes
\[
 k_m=8\cos(m\pi/320),\qquad 0\le m\le320.
\]
We approximate $\gamma$ by its degree-320 Chebyshev interpolant on $[-8,8]$ and set it to zero outside this interval. In the evaluation of the asymptotic formula below, $\gamma$ denotes this numerical approximation. Figure~\ref{fig:numerical_scattering} shows the initial amplitudes and the computed reflection coefficient.

\begin{figure}[htbp]
 \centering
 \begin{minipage}[t]{.485\linewidth}
  \centering\includegraphics[width=\linewidth]{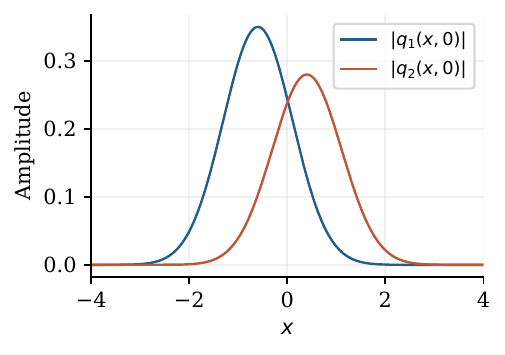}
  {\small (a) Initial amplitudes.\par}
 \end{minipage}\hfill
 \begin{minipage}[t]{.485\linewidth}
  \centering\includegraphics[width=\linewidth]{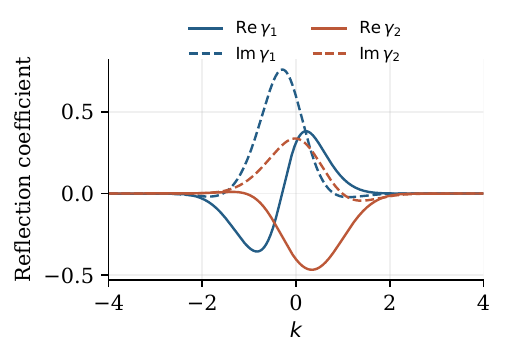}
  {\small (b) Reflection coefficient.\par}
 \end{minipage}
 \caption{(a) Initial amplitudes $|q_1(x,0)|$ and $|q_2(x,0)|$. (b) Real and imaginary parts of the computed reflection coefficient $\gamma=(\gamma_1,\gamma_2)$.}
 \label{fig:numerical_scattering}
\end{figure}

We examine convergence at $k=0$ as $N$ increases. Let $\gamma_j^{(N)}(0)$ denote the computed value for component $j$, using $N$ nodes on each half-line. With the $N=301$ values as a reference, Figure~\ref{fig:direct_convergence} plots
\[
 \bigl|\gamma_j^{(N)}(0)-\gamma_j^{(301)}(0)\bigr|, \qquad j=1,2,\qquad N=10,15,\ldots,300.
\]
The differences decrease to approximately $10^{-12}$ and then level off.

\begin{figure}[htbp]
 \centering
 \begin{minipage}[t]{.485\linewidth}
  \centering\includegraphics[width=\linewidth]{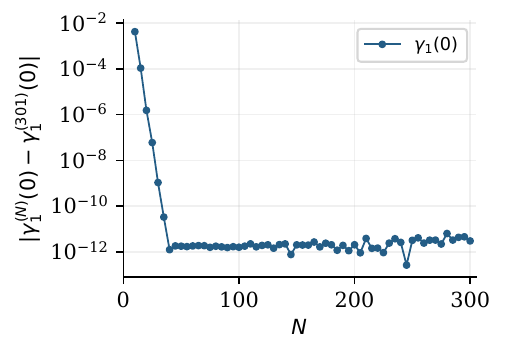}
  {\small (a) $\gamma_1(0)$.\par}
 \end{minipage}\hfill
 \begin{minipage}[t]{.485\linewidth}
  \centering\includegraphics[width=\linewidth]{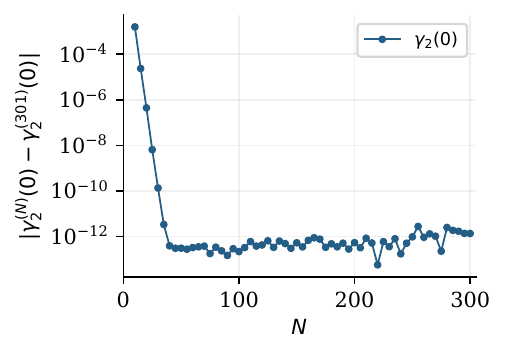}
  {\small (b) $\gamma_2(0)$.\par}
 \end{minipage}
 \caption{Convergence of the reflection coefficient at $k=0$. Each panel shows $|\gamma_j^{(N)}(0)-\gamma_j^{(301)}(0)|$, with $N$ Chebyshev--Lobatto nodes on each half-line.}
 \label{fig:direct_convergence}
\end{figure}

\subsection{Evaluation of the asymptotic formula}

\begin{figure}[htbp]
 \centering
 \begin{minipage}[t]{\linewidth}
  \centering\includegraphics[width=\linewidth]{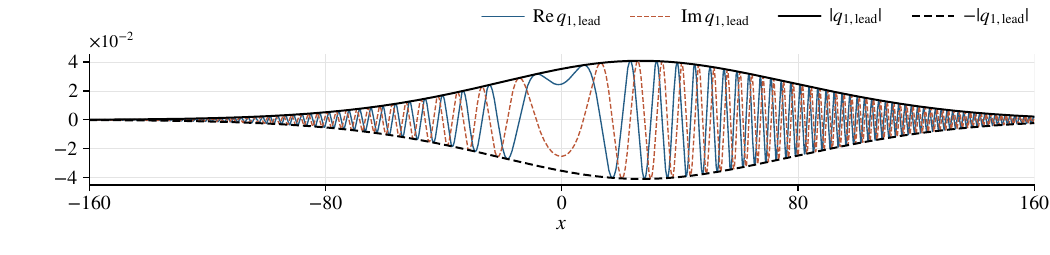}
  {\small (a) $q_{1,\mathrm{lead}}$, $t=20$.\par}
 \end{minipage}
 \par\vspace{1mm}
 \begin{minipage}[t]{\linewidth}
  \centering\includegraphics[width=\linewidth]{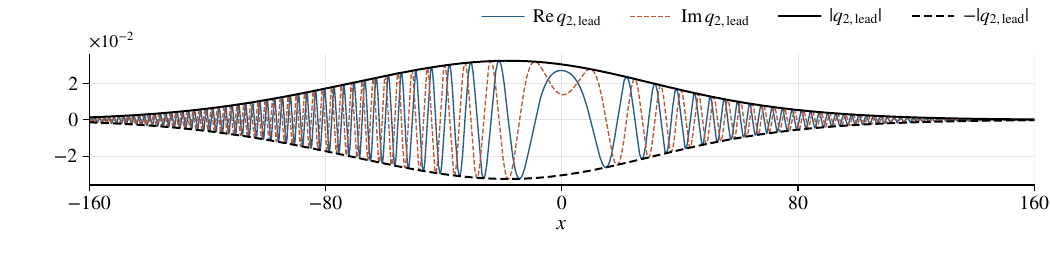}
  {\small (b) $q_{2,\mathrm{lead}}$, $t=20$.\par}
 \end{minipage}
 \par\vspace{1mm}
 \begin{minipage}[t]{\linewidth}
  \centering\includegraphics[width=\linewidth]{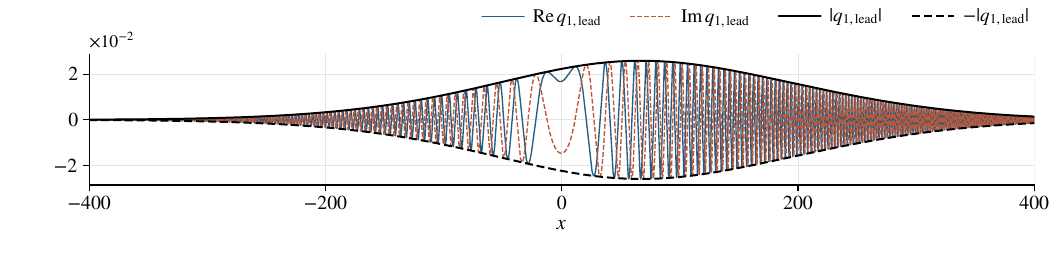}
  {\small (c) $q_{1,\mathrm{lead}}$, $t=50$.\par}
 \end{minipage}
 \par\vspace{1mm}
 \begin{minipage}[t]{\linewidth}
  \centering\includegraphics[width=\linewidth]{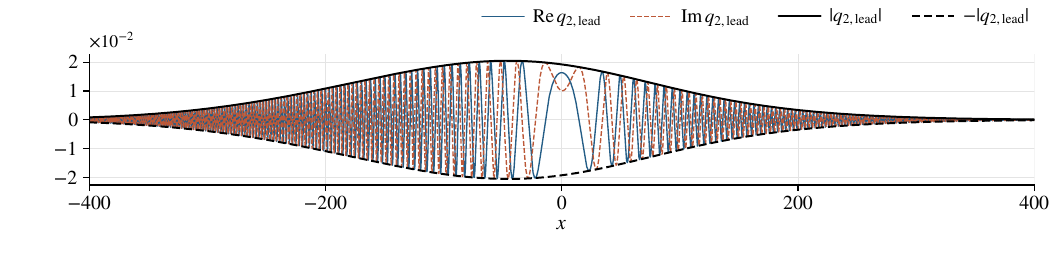}
  {\small (d) $q_{2,\mathrm{lead}}$, $t=50$.\par}
 \end{minipage}
 \par
 \caption{The asymptotic solution $\qq_{\mathrm{lead}}$ for $\lambda=0.05$ at $t=20$ and $t=50$. Blue and orange curves show the real and imaginary parts of $q_{j,\mathrm{lead}}$, respectively. Black solid curves show $|q_{j,\mathrm{lead}}|$, and black dashed curves show $-|q_{j,\mathrm{lead}}|$. The spatial intervals are $[-160,160]$ at $t=20$ and $[-400,400]$ at $t=50$.}\label{fig:complex_asymptotics}
\end{figure}

For $k_0\in[-3,3]$, we compute $\nu(k_0)$ and $A_0(k_0)$ from \eqref{eq:endpoint_data}, taking $A_0(k_0)=0$ when $\gamma(k_0)=0$. Using the truncated reflection coefficient, we compute $\chi_0(k_0)$ and $I_0(k_0)$ from
\[
 \begin{aligned}
 \chi_0(k_0)&\approx\frac1{2\pi\ii}\biggl[\int_{-8}^{k_0-1}\frac{\log(1+|\gamma(\xi)|^2)}{\xi-k_0}\,\dd\xi+\int_{k_0-1}^{k_0}\frac1{\xi-k_0}\log\frac{1+|\gamma(\xi)|^2}{1+|\gamma(k_0)|^2}\,\dd\xi\biggr],\\
 I_0(k_0)&\approx\frac1\pi\int_{k_0}^{8}\log(1+|\gamma(\xi)|^2)\,\dd\xi.
 \end{aligned}
\]
We evaluate each integral with a 192-point Gauss--Legendre rule. In the second integral defining $\chi_0$, the endpoint singularity is removable. The integrand has the limit $2\re(\gamma'(k_0)\gamma(k_0)^\dagger)/(1+|\gamma(k_0)|^2)$ as $\xi\to k_0$.

To compute $E_0(k_0)$, we factor out the endpoint oscillation of $\delta$ while preserving its normalization at infinity. For fixed $k_0$, set
\[
 \mathcal B(k;k_0)=\exp\!\left[-\ii A_0(k_0)e^{k-k_0}\mathrm E_1(k-k_0)\right],\qquad \mathcal F(k;k_0)=\mathcal B(k;k_0)^{-1}\delta(k;k_0),
\]
where $\mathrm E_1$ is the principal-branch exponential integral, obtained by analytic continuation of $\int_z^\infty e^{-\xi}\xi^{-1}\,\dd\xi$ from $z>0$ to $\mathbb C\setminus(-\infty,0]$. Near zero, $\mathrm E_1(z)=-\gamma_{\mathrm E}-\log z+O(z)$, where $\gamma_{\mathrm E}$ is the Euler--Mascheroni constant. Comparing $\mathcal B$ with $\mathcal D(k)=\exp[\ii A_0(k_0)\log(k-k_0)]$ from Lemma~\ref{lem:endpoint} gives
\[
 \lim_{k\downarrow k_0}\mathcal D(k)^{-1}\mathcal B(k;k_0)=\exp[\ii\gamma_{\mathrm E}A_0(k_0)].
\]
Thus, Lemma~\ref{lem:endpoint} implies that $\mathcal F$ has a finite limit as $k\downarrow k_0$. Since $\mathcal B(k;k_0)\to\Id_2$ as $k\to\infty$, $\mathcal F$ also retains the normalization $\mathcal F(\infty;k_0)=\Id_2$.

For $r>0$, the boundary values of $\mathcal B$ are
\[
 \mathcal B_\pm(k_0-r;k_0)=\exp\!\left[\ii A_0(k_0)e^{-r}\mathrm{Ei}(r)\mp\pi A_0(k_0)e^{-r}\right],
\]
where $\mathrm{Ei}(r)=-\re\mathrm E_1(-r+\ii0)$ is real. The transformed RH problem is
\[
 \begin{aligned}
 &\mathcal F_+=J_{\mathcal F}\mathcal F_-,\qquad\mathcal F(\infty;k_0)=\Id_2,\\
 &J_{\mathcal F}(\xi;k_0)=\mathcal B_+(\xi;k_0)^{-1}(\Id_2+\gamma(\xi)^\dagger\gamma(\xi))\mathcal B_-(\xi;k_0).
 \end{aligned}
\]
Following the framework of \cite[Section~6.2]{TrogdonOlver2016}, we use piecewise polynomial collocation to solve the singular integral equation
\[
 \mathcal F_-=\Id_2+C_-\!\left[(J_{\mathcal F}-\Id_2)\mathcal F_-\right].
\]
Here, $C_-$ is the Cauchy boundary operator on $(-\infty,k_0]$, oriented from left to right. The regularized jump satisfies $J_{\mathcal F}(k_0-r;k_0)-\Id_2=O(r)$ as $r\downarrow0$, and $\mathcal F_-$ remains bounded by Lemma~\ref{lem:endpoint}.

Set $r=k_0-\xi$. For the truncated reflection coefficient used here, $J_{\mathcal F}(k_0-r;k_0)-\Id_2$ decays exponentially as $r\to\infty$. We truncate the integral equation at $r=40$ and divide $[0,40]$ at
\[
 0,\ 10^{-7},\ 10^{-6},\ \ldots,\ 10^{-1},\ \tfrac14,\ \tfrac12,\ 1,\ 2,\ 4,\ 8,\ 16,\ 24,\ 32,\ 40.
\]
On each of the 17 panels, we interpolate the density $(J_{\mathcal F}-\Id_2)\mathcal F_-$ at 16 Gauss--Legendre nodes and collocate the integral equation at the same nodes. Let $r_i$, $1\le i\le272$, denote the collocation nodes, and let $\ell_j$ be the Lagrange basis function associated with $r_j$. It is a polynomial of degree at most 15 on the panel containing $r_j$, vanishes on the other panels, and satisfies $\ell_j(r_i)=\mathbf1_{\{i=j\}}$. Integrating these basis functions against the Cauchy kernel gives the matrix entries
\[
 (\mathsf C_-)_{ij}=-\frac1{2\pi\ii}\operatorname{pv}\!\int_0^{40}\frac{\ell_j(r)}{r-r_i}\,\dd r-\frac12\mathbf1_{\{i=j\}}.
\]
Here, $\operatorname{pv}$ denotes the Cauchy principal value, and $\mathbf1_{\{i=j\}}$ is one for $i=j$ and zero otherwise. The substitution $r=k_0-\xi$ gives the factor $-1/(2\pi\ii)$, while the term $-\tfrac12\mathbf1_{\{i=j\}}$ comes from the lower boundary value in the Plemelj formula. On a panel $[a,b]$ containing both $r_i$ and $r_j$,
\[
 \operatorname{pv}\!\int_a^b\frac{\ell_j(r)}{r-r_i}\,\dd r=\int_a^b\frac{\ell_j(r)-\ell_j(r_i)}{r-r_i}\,\dd r+\ell_j(r_i)\log\frac{b-r_i}{r_i-a}.
\]
The first integrand is a polynomial of degree at most 14, with its value at $r_i$ defined by continuity. The integrals over the other panels are nonsingular. Collocation gives the linear system
\[
 \mathcal F_-(k_0-r_i;k_0)-\sum_j(\mathsf C_-)_{ij}[J_{\mathcal F}(k_0-r_j;k_0)-\Id_2]\mathcal F_-(k_0-r_j;k_0)=\Id_2.
\]
The 272 collocation nodes yield a $544\times544$ linear system with two right-hand sides, one for each column of $\mathcal F_-$. We recover the endpoint limit $\mathcal F(k_0;k_0)=\lim_{k\downarrow k_0}\mathcal F(k;k_0)$ by applying the 16-point Gauss--Legendre rule on each panel to the integral below. The relation between $\mathcal B$ and $\mathcal D$ then gives $E_0(k_0)$:
\[
 \begin{aligned}
 &\mathcal F(k_0;k_0)\approx\Id_2-\frac1{2\pi\ii}\int_0^{40}\frac{[J_{\mathcal F}(k_0-r;k_0)-\Id_2]\mathcal F_-(k_0-r;k_0)}{r}\,\dd r,\\
 &E_0(k_0)=\exp[\ii\gamma_{\mathrm E}A_0(k_0)]\mathcal F(k_0;k_0).
 \end{aligned}
\]
The numerator is $O(r)$, so the endpoint integrand remains bounded as $r\downarrow0$. When $\gamma(k_0)=0$, we have $A_0(k_0)=0$ and $\mathcal B\equiv\Id_2$, and the same procedure applies.

We tabulate $E_0$, $\chi_0$, and $I_0$ at the 257 Chebyshev--Lobatto nodes $k_{0,n}=3\cos(n\pi/256)$, $0\le n\le256$. We use degree-256 Chebyshev interpolation between these nodes, applied entrywise to $E_0$. The parameter $\nu$ is evaluated directly from the interpolated reflection coefficient.

At the 257 tabulation nodes, the maximum values of $|E_0^\dagger E_0-\Id_2|$ and $|\det E_0-e^{\chi_0}|$ are $3.16\times10^{-9}$ and $1.32\times10^{-9}$, respectively.

Evaluating the computed quantities at $k_0=-x/(4t)$ gives $\qq_{\mathrm{lead}}$. Figure~\ref{fig:complex_asymptotics} shows both components at $t=20$ and $t=50$ on $[-8t,8t]$, corresponding to $k_0\in[-2,2]$. At fixed $x/t$, the component envelopes scale as $t^{-1/2}$. For fixed CNLS initial data $\ww_0$, $\lambda$ enters only through the phase factor $e^{\ii\lambda I_0(k_0)}$ and leaves the envelopes unchanged.

\medskip
\noindent \textbf{Acknowledgements} The authors thank Dr. Shikun Cui for valuable suggestions. This work is supported by the Kaiyuan International Mathematical Sciences Institute, the Fundamental Research Funds for the Central Universities (JKF2025077492962), and the Academic Excellence Foundation of BUAA for Ph.D. Students.

During the preparation of this manuscript, the authors used OpenAI GPT-6.1 Sol and GPT-6 Astra to assist with mathematical derivations. They reviewed, verified, and edited all AI-generated content and take full responsibility for its correctness and for the final manuscript.

\section*{Declarations}
\noindent \textbf{Conflict of interest} The authors declare that they have no conflict of interest.

\noindent\textbf{Data availability} The computational materials supporting this work are not publicly available. Requests for further information should be addressed to the authors.

\end{sloppypar}
\end{document}